\documentclass[smallextended]{svjour3}       
\RequirePackage{fix-cm}
\smartqed  

\usepackage{amsmath}
\usepackage{amssymb}
\usepackage{amsfonts}

\newcommand{\handout}[5]{
   \renewcommand{\thepage}{#1-\arabic{page}}
   \noindent
   \begin{center}
   \framebox{
      \vbox{
    \hbox to 5.78in { {\bf CSC 2414 Lattices in Computer Science}
     	 \hfill #2 }
       \vspace{4mm}
       \hbox to 5.78in { {\Large \hfill #5  \hfill} }
       \vspace{2mm}
       \hbox to 5.78in { {\it #3 \hfill #4} }
      }
   }
   \end{center}
   \vspace*{4mm}
}

\newcommand{\comment}[1]{}

\DeclareMathOperator*{\argmax}{arg\,max}
\DeclareMathOperator*{\argmin}{arg\,min}

\newenvironment{proof-sketch}{\noindent{\bf Sketch of Proof}\hspace*{1em}}{\qed\bigskip}
\newenvironment{proof-idea}{\noindent{\bf Proof Idea}\hspace*{1em}}{\qed\bigskip}
\newenvironment{proof-of-corollary}[1]{\noindent{\bf Proof of Lemma #1}\hspace*{1em}}{\qed\bigskip}
\newenvironment{proof-attempt}{\noindent{\bf Proof Attempt}\hspace*{1em}}{\qed\bigskip}

\makeatletter
\def\fnum@figure{{\bf Figure \thefigure}}
\def\fnum@table{{\bf Table \thetable}}
\long\def\@mycaption#1[#2]#3{\addcontentsline{\csname
  ext@#1\endcsname}{#1}{\protect\numberline{\csname
  the#1\endcsname}{\ignorespaces #2}}\par
  \begingroup
    \@parboxrestore
    \small
    \@makecaption{\csname fnum@#1\endcsname}{\ignorespaces #3}\par
  \endgroup}
\def\mycaption{\refstepcounter\@captype \@dblarg{\@mycaption\@captype}}
\makeatother

\newcommand{\mathify}[1]{\ifmmode{#1}\else\mbox{$#1$}\fi}
\newcommand{\bigO}O

\newcommand{\remove}[1]{}
\newcommand{\ignore}[1]{}

\def\R{\mathbb{R}}

\newcommand{\alg}{\mathcal{A}}

\usepackage[all]{xy} 
\usepackage{subfig,color} 
\usepackage{listings}
\usepackage{fullpage}
\usepackage{ifpdf}

\makeatletter 
\newif\if@restonecol 
\makeatother   
\usepackage[ruled,linesnumbered]{algorithm2e} 
\newcommand{\reals}{\ensuremath{\mathbb R}} 
 
\newcommand{\be}{
\begin{equation}
} 
\newcommand{\ee}{
\end{equation}
}

\newcommand{\mb}[1]{\mathbf{#1}}

\DeclareCaptionType{copyrightbox}

\comment{
\makeatletter 
\renewcommand\bibsection

{ 
\section*{\refname \@mkboth{\MakeUppercase{\refname}}{\MakeUppercase{\refname}}} } 
\makeatother}

\begin{document}

\title{Strategyproof Mechanisms for Competitive Influence in Networks}
\author{Allan Borodin
\and Mark Braverman
\and
Brendan Lucier
\and
Joel Oren
}

\institute{Allan Borodin \and Joel Oren \at Department of Computer Science, University of Toronto, 10 King's College Road, Toronto, Ontario, M5S 3G4, Canada. Fax: +1-416-978-1931 \\\email{\{bor,oren\}@cs.toronto.edu}
\and Mark Braverman \at Department of Computer Science, Princeton University, 35 Olden Street, Princeton, NJ 08540, USA. \\\email{mbraverm@cs.princeton.edu}
\and Brendan Lucier \at Microsoft Research, 1 Memorial Drive, Cambridge, MA 02142, USA.\\\email{brlucier@microsoft.com}}

\maketitle

\begin{abstract}
Motivated by applications to word-of-mouth advertising, we consider a game-theoretic scenario in which competing advertisers want to target initial adopters in a social network.
Each advertiser wishes to maximize 
the resulting cascade of influence, modeled by a general network diffusion process.
However, competition between products may adversely impact the rate of adoption for any given firm.
The resulting framework gives rise to complex preferences that depend on the specifics of the stochastic diffusion model and the network topology.

We study this model from the perspective of a central mechanism, such as 
a social networking platform, that can optimize seed placement as a service for the advertisers.  We ask: given the reported demands of the competing firms, how should a mechanism choose seeds to maximize overall efficiency?
Beyond the algorithmic problem, 
competition raises issues of strategic behaviour: rational agents should not be incentivized to underreport their budget demands.

We show that when there are two players, the social welfare can be $2$-approximated by a polynomial-time strategyproof mechanism.  Our mechanism is defined recursively, randomizing the order in which advertisers are allocated seeds according to a particular greedy method.  For three or more players, we demonstrate that under additional assumptions (satisfied by many existing models of influence spread) there exists a simpler strategyproof $\frac{e}{e-1}$-approximation mechanism; notably, this second mechanism is not necessarily strategyproof when there are only two players. 
\end{abstract}
\smallskip

\section{Introduction}
The concept of word-of-mouth advertising is built upon the idea that referrals between individuals can lead to a contagion of opinion in a population.  In this way, a small number of initial adopters can generate a cascade of influence, significantly impacting the adoption of a new product.  While this concept has been very well studied in the marketing and sociology literature \cite{G78,S78,BR87,GLM01,CM07}, 
recent popularity of online social networking has made it possible to obtain rich data and directly target individuals based on network topology.
Indeed, a potential advantage of advertising served via online social networks is that the platform could preferentially target central individuals, impacting the overall effectiveness of its advertisers' campaigns.

Various models of network influence spread have arisen recently in the literature, with a focus on the algorithmic problem of deciding which individuals to target as initial adopters (or ``seeds")
\cite{Kempe:2003,Kempe:2005,DBLP:journals/siamcomp/MosselR10}.  One commonality among many of these 
(stochastic) models is that the expected number of eventual adopters is a non-decreasing submodular function of the seed set.  This implies that natural greedy methods \cite{Nemhauser1978} can be used to choose initial adopters to approximately maximize an advertiser's expected influence.  Of course, actually applying such algorithms requires intimate knowledge of the social network, which may not be readily available to all advertisers.  However, the owners of the network data (e.g.\ Facebook or Google) could more easily find potentially influential individuals to target.  Our goal is to study the problem faced by a network platform who wishes to provide this service to its advertisers.

Consider the following framework.  An online social network platform sells advertising space by contract, offering a price per impression to advertising firms.  Each firm has an advertising budget, which determines a number of ad impressions they wish to display.  As an additional service to the firms, the platform attempts to optimize the placement of advertisements so to maximize influence diffusion.  This optimization is to be provided as a service to the advertisers, with the primary goal of making the social network more attractive as a marketing platform.  The network provider thus faces an algorithmic problem: maximize the total influence of the advertisers given their demands (i.e.\ number of impressions).  This problem may be complicated by competition between advertisers, which results in negative externalities upon each others' product adoptions.  Moreover, since advertising budgets are private, there is also a game-theoretic component to the problem: the placement algorithm should not incentivize firms to reduce their budgets.  This may happen if, due to eccentricities of the algorithm, lower-budget advertisers might obtain higher expected influence than advertisers with higher budgets.

Crucial to this problem formulation is the way in which influence is modeled by the advertising platform.  We present a general submodular assignment problem with negative externalities, which captures most previous influence models that have been proposed in the literature \cite{Carnes:2007:MIC:1282100.1282167,DBLP:conf/wine/2007/BharathiKS07,springerlink:10.1007/978-3-642-17572-5_48,GoyalK12}.  Within this framework, we consider the optimization problem faced by a central mechanism that must determine the seed nodes for \emph{each} advertiser, given the 
advertisers' budget constraints. The goal of the mechanism is to maximize the overall efficiency of the marketing campaigns, but the advertisers are strategic and may underreport their budget demands to increase their own product adoption rates.


Two points of clarification are in order.  First, our formulation differs from a line of prior work that studies equilibria of the game in which each advertiser selects their seed set directly \cite{GoyalK12,DBLP:conf/wine/2007/BharathiKS07}.  Such a game supposes that each advertiser has detailed knowledge of the social network topology, the ability to compute or converge to equilibrium strategies, and the power to target arbitrary individuals in the network.  Our work differs in that we assume that the targeted advertising goes through an intermediary (the social network), which selects seed sets on the players' behalf.  

Second, we suppose that advertiser budgets and the price per impression are set exogenously (or, alternatively, that the seeds correspond to special offers or other interventions of limited quantity).  As such, we do not explicitly model the problem of maximizing revenue; rather, the role of our mechanism is to decide where to place the purchased impressions.  In this sense our framework is closer in spirit to matching algorithms for display advertising \cite{FMMM09,FHKMS10} than to revenue-optimal mechanism design.
There are many ways in which this model could be enriched, such as by endogenizing budgets or allowing complex pricing schemes that depend upon expected influence.  We leave these as avenues for future work, though we note that such extensions presuppose that agents have sufficient knowledge of the spread process and graph topology to accurately value initial adoption sets.


Our model of competitive influence spread is described formally in Section \ref{sec.preliminaries}. 
Our formulation captures and extends many existing models of influence spread, allowing incorporation of features such as node weights, player-specific spread probabilities, and non-linear selection probabilities.  A more detailed discussion appears in Appendix \ref{sec.models}.


We wish to design mechanisms that are strategyproof, in that rational agents are incentivized to truthfully reveal their demands.  In particular, an agent should not be able to increase its expected influence by reducing its requested number of seeds (i.e.\ budget).
The difficulty in avoiding such non-monotonicities is that the expected outcome of an advertiser can be negatively impacted by externalities imposed by the allocation to its opponents, which can depend on the budget declarations in a non-trivial manner.
\bigskip

\noindent
\textbf{Our Results: } 
We design three different strategyproof mechanisms for the competitive influence maximization problem, for use in varying circumstances.  Our main result is a $2$-approximate strategyproof mechanism for use when there are two competing advertisers, under a very general model of influence spread.  
This mechanism uses a novel technique for monotonizing the expected utilities of the agents using geometric properties of the problem in the two-player case.  

Our construction is based upon a greedy algorithm for submodular function maximization subject to a partition matroid constraint, known as the locally greedy algorithm \cite{Nemhauser1978,Goundan2007}.   This algorithm repeatedly chooses an agent in each round, and assigns a node to that agent in order to maximize the marginal increase to social welfare.  As we discuss in Section \ref{sec.examples}, this algorithm is not strategyproof in general.  However, it has the property that the choice of agent in each round is arbitrary; this provides a degree of freedom that can be exploited to obtain strategyproofness.  Indeed, for the case of two agents, we show how to recursively construct a distribution over potential allocations returned by locally greedy algorithms, with the property that each agent's expected individual value under this distribution is monotone\footnote{We use the word \emph{monotone} in its game-theoretic sense, meaning that a player's outcome is a monotone function of its bid. We distinguish this from the monotonicity of the social welfare function of the mechanism, and use the term \emph{non-decreasing} when referring to the social welfare function.} with respect to the number of initial elements allocated.

Our second mechanism is for three or more players, under some natural restrictions on the influence spread process.  Specifically, we require two properties:  first, the social welfare is independent of the manner in which elements are partitioned among the players (\emph{mechanism indifference}).  Second, the payoff of a player does not depend on the manner in which the elements allocated to her competitors are partitioned among the competitors (\emph{agent indifference}).  These conditions are defined formally in Section \ref{sec:3player}.  
We note that these assumptions are implicit in many prior models of influence spread \cite{Carnes:2007:MIC:1282100.1282167,DBLP:conf/wine/2007/BharathiKS07}.  Under these assumptions, we develop a strategyproof mechanism that obtains a $\frac{e}{e-1}$-approximation to the optimal social welfare when there are three or more players. 
Interestingly, our analysis makes crucial use of the presence of three or more players, and indeed we show that this mechanism fails to be strategyproof when only two players are present, even with these two additional assumptions\footnote{Notice that the agent indifference property 
holds vacuously in the two-player case, as there is only one other player.}.

Our final mechanism construction 
satisfies an additional constraint that agent allocations be disjoint.  In its most general form, our problem specification does not require that the set of elements allocated to the agents be disjoint\footnote{Many prior models of competitive influence do allow non-disjoint allocations \cite{GoyalK12,DBLP:conf/wine/2007/BharathiKS07}; our intention is to demonstrate that a disjointness condition can be accommodated if necessary, rather than imply that non-disjointness is undesirable.}. Our first mechanism described above may place a given node in the seed sets of multiple players. Our second mechanism for more than two players 
produces a disjoint allocation
when the greedy algorithm used for a single 
player results in a disjoint allocation. 
When it is desirable for allocations to be disjoint, we show how our construction can be modified to work under this additional requirement, resulting in a strategyproof $3$-approximation mechanism.  This result requires that we impose a symmetry assumption on the influence spread model, which states that the outcome of the influence process is invariant under relabelling of the players\footnote{We note that this property holds for most models of influence spread studied in the literature \cite{GoyalK12,DBLP:conf/wine/2007/BharathiKS07,Carnes:2007:MIC:1282100.1282167,KM2011}.}.  
Our mechanisms run in time polynomial in the demands submitted by the agents and in the size of the underlying ground set.  This dependence on the demand values is necessary, as the mechanism constructs a solution consisting of sets of this size.  Our dependence on the size of the underlying ground set is captured by queries for an element that maximizes a marginal increase in social welfare.  Given oracle access to queries of this nature, our algorithm would run in time polynomial in the declared demands.
Generally speaking, the spread process itself is randomized and as in \cite{Kempe:2003,Kempe:2005}, the oracle can be viewed as providing an element that approximately maximizes
the marginal gain 
by sampling enough trials of this process \cite{Kempe:2003,Kempe:2005}.  
Our analysis also holds when such approximate marginal maximizers are used to implement our underlying greedy algorithm; following the exposition in \cite{Goundan2007}, such an approximate maximizer provides an approximation
that approaches 2 as the oracle approximation approaches 1. We will simplify our discussion throughout
by assuming it is possible to find elements that exactly maximize marginal gains in social welfare.  

\bigskip

\noindent
\textbf{Related Work:}
%
Models of influence spread in networks, covering both cascade and threshold phenomena, are well-studied in the sociology and marketing literature \cite{G78,S78}.
The (non-competitive) problem of maximizing influence in social networks was 
theoretically modelled by Kempe et al. \cite{Kempe:2003,Kempe:2005}. Subsequent papers extended these models to a competitive setting in which there are multiple advertisers. Carnes et al. \cite{Carnes:2007:MIC:1282100.1282167} suggested the Wave Propagation model and the Distance Based model, which were based on the Independent Cascade model. Additionally, Dubey et al. \cite{DBLP:conf/wine/DubeyGM06}, Bharathi et al. \cite{DBLP:conf/wine/2007/BharathiKS07}, Kosta et al. \cite{KostkaOW08}, and Apt et al. \cite{KM2011} also studied various competitive models. The main issue that these models addressed was how to arbitrate ties in each step of the process,
determining which technology a node will assume when reached by several technologies at once. The main algorithmic task addressed by these models is choosing the optimal set of nodes for a player entering an existing market, in which the competitor's choice of initial nodes is already known. Borodin et al. \cite{springerlink:10.1007/978-3-642-17572-5_48} presented the OR model which proposes a different approach, in which the previously studied, non-competitive diffusion models proceed independently for each technology as a first 
phase of the process, after which the nodes decide between each technology according to some decision function. 

Independent of our work \cite{BBLO12}, Goyal and Kearns \cite{GoyalK12} provided bounds on the efficiency of equilibria (i.e. the price of anarchy) in a competitive influence game played by two players.  Their influence spread model is characterized by switching functions (specifying the process by which a node decides to adopt a product) and selection functions (specifying the manner in which nodes decide which product to adopt).
They demonstrate that an equilibrium of the resulting game yields half of the optimal social welfare, given that the switching functions are concave.
Their model is closely related to our own.  
Specifically, the social welfare function is monotone and satisfies the mechanism indifference assumption, and concavity of the switching function implies that the social welfare is submodular (by \cite{DBLP:journals/siamcomp/MosselR10}), so our mechanism for two players applies to their model as well\footnote{An ``adverse competition'' assumption in \cite{GoyalK12} is stated 
for $k = 2$ agents and holds at every node. Their assumption is somewhat 
weaker than ours, 
which we only apply to
the social welfare function. See section ~\ref{sec.preliminaries}.}.  Goyal and Kearns also note that their results extend to $k > 2$ players, resulting in an approximation factor of $2k$, when the selection function is linear; this linearity implies our agent indifference assumption, and hence our mechanism for three or more players also applies.
However, we note that the
Goyal and Kearns results on efficiency at equilibrium are
satisfied without an intervening mechanism and hence are incomparable with the  mechanism results of this paper. 

Finally, to the best of our knowledge, there is only one other paper that
considers a mechanism design problem in the context of competitive
influence spread. Namely, Singer \cite{Singer12} considers a social
network where the nodes are viewed as agents who have private costs
for hosting a product and the mechanism has a budget for inducing some set of
initial nodes to become hosts. The mechanism wishes to maximize the number
of nodes that will eventually be influenced and each agent wishes
to maximize their profit equal to the inducement received minus its private
cost.



 \section{Preliminaries}
\label{sec.preliminaries}
We consider a setting in which there is a ground set $U = \{e_1, \dotsc, e_n\}$ of $n$ elements (e.g.\ nodes in a social network), and $k$ players.
An allocation is some $(S_1,\ldots, S_k) \in 2^U \times \cdots \times 2^U$; that is, an assignment of set\footnote{For notational convenience we will assume that $S_1,\ldots, S_k$ are sets, but our results extend to permit multisets (i.e., where the same element can be awarded multiple times to one agent).
See Appendix~\ref{sec.models} for further discussion.} 
$S_i$ to each player $i$.  
For the most part we will follow the convention that these sets should be disjoint, though in general our model does not require disjointness.  In particular, we consider a setting in which sets need not be disjoint in Section \ref{sec.main}.

We are given functions $f_i \colon 2^U \times \cdots 2^U \to \reals_{\geq 0}$, denoting the expected values of players $i=1,\ldots,k$, for allocation $(S_1,\ldots,S_k)$. We define 
$f = \sum_{i=1}^k f_i$, so that $f(\mathbf{S})=f(S_i,\mathbf{S_{-i}})$ denotes the total expected welfare of the allocation $(\mathbf{S})=(S_1,\ldots,S_k)=(S_i, \mathbf{S_{-i}})$.

We will require that functions $f$, and $f_1,\ldots, f_k$ satisfy certain properties, motivated by known properties of influence spread models studied in the literature.  First, we will assume that $f$ is a submodular non-decreasing function, in the following sense.  For any $S_i \subseteq S'_i$, $\mathbf{S_{-i}}$, and $e \in U$, we have
$f(S_i,\mathbf{S_{-i}}) \leq f(S'_i,\mathbf{S_{-i}})$ and
$$f(S_i \cup \{e\},\mathbf{S_{-i}}) - f(S_i,\mathbf{S_{-i}}) \geq f(S'_i \cup \{e\}, \mathbf{S_{-i}}) - f(S'_i,\mathbf{S_{-i}}).$$
We will also require that for all $i=1,\ldots,k$, the function $f_i$ be non-decreasing in the allocation to player $i$, so that $f_i(S_i,\mathbf{S_{-i}}) \leq f_i(S'_i,\mathbf{S_{-i}})$ for any $S_i \subseteq S'_i$.

We impose one final model assumption,
which we call \emph{adverse competition}: that each $f_i$ is non-increasing in the allocation to other players.  That is, for all $j \neq i$, $f_i(S_j,\mathbf{S_{-j}}) \geq f_i(S'_j,\mathbf{S_{-j}})$ for any $S_j \subseteq S'_j$. This assumption captures our intuition that, in a competitive influence model, the presence of additional adopters for one player can only impede the spread of influence for another player.
We discuss the motivation for and necessity of this assumption in Appendix \ref{sec.models}. 

We study the following algorithmic problem.  Given input values $b_1,\ldots,b_k \geq 0$, we wish to find sets $S_1,\ldots, S_k \subseteq U$, with $|S_i| = b_i$, for all $i=1,\ldots,k$, such that $f(S_1,\ldots,S_k)$ is maximized.  We assume we are given oracle access to the functions $f$ and $f_1, \ldots, f_k$.  
Note that we impose a ``demand satisfaction'' condition on the mechanism, that each agent is allocated all of his demand. (To this end we will assume that $|U| \geq \sum_{i=1}^kb_i$; i.e.\ that there are enough items to allocate).

Suppose that $\alg$ is a deterministic algorithm for the above problem, so that $\alg(b_1,\ldots,b_k)$ denotes an allocation for any $b_1,\ldots,b_k \geq 0$.  We say that $\alg$ is \emph{monotone} if, for all bid vectors $\mathbf{b}=(b_1,\ldots,b_k) \in \mathbb{Z}^k_{\geq 0}$, $f_i(\alg(b_i,\mathbf{b_{-i}})) \leq f_i(\alg(b_i+1,\mathbf{b_{-i}}))$, for each player $i=1,\ldots,k$.  We extend this definition to randomized algorithms in the natural way, by taking expectations over the outcomes returned by $\alg$.

We will assume that each player $i$ has a \emph{type} $\tilde{b_i}$, representing the maximum number of elements they can be allocated.  
The utility of player $i$ for allocation $\mathbf{S}=(S_1,\ldots,S_k)$ is 
\[ u_i(\mathbf{S}) = \begin{cases}f_i(S_i,\mathbf{S_{-i}}) \quad & \text{if $|S_i| \leq \tilde{b_i}$}\\ -\infty & \text{otherwise.}\end{cases} \]
We then say that algorithm $\alg$ is \emph{strategyproof} if, for all $\mathbf{b} \in \mathbb{Z}^k_{\geq 0}$ and $b_i' \leq b_i$, $u_i(\alg(b_i',\mathbf{b_{-i}})) \leq u_i(\alg(b_i, \mathbf{b_{-i}}))$.  In other words, an algorithm is strategyproof if it incentivizes each agent to report its type truthfully.  

The problem of maximizing welfare function $f(\cdot)$ subject to the reported demands can be stated in the framework of maximizing a submodular set-function subject to a \emph{partition matroid} constraint. An instance of a partition matroid $\mathcal{M}=(E,\mathcal{F})$ is given by a union of disjoint sets $E=\bigcup_{i=1,\ldots,k}E_i$, and a set of corresponding cardinality constraints $d_1,\ldots,d_k$. A set $X$ is in $\mathcal{F}$, i.e. is \emph{independent}, if $|X \cap E_i| \leq d_i$, for all $1 \leq i \leq k$. That is, an independent set is formed by taking no more than the prescribed size constraint for each of the sets. The optimization problem to find an independent set that maximizes a non-decreasing and submodular set-function $g : \mathcal{F} \rightarrow \reals_{\geq 0}$. Our problem falls into this framework by setting the ground set to be $U \times \{1,\ldots,k\}$, the cardinality constraints $d_i=b_i$, for all $i$ and setting the objective function to be the social welfare:
\begin{equation}
g(X)=f(\mathbf{S}), \text{ where }X=\bigcup_{i=1}^k(S_i \times \{i\}).
\end{equation}
We note, however, that this formulation does not apply if the allocated sets are required to be disjoint. The addition of disjointness causes our constraint to no longer take the form of a matroid, an issue which will be addressed in Section ~\ref{sec.disjoint}. Also note that this alternative definition of our setting conforms to the single-parameter convention of submodular set-functions. However, we will mostly refer to the former formulation of the problem for clarity and succinctness.

As a result of this correspondence with the framework of partition matroids, we will be interested in a particular greedy algorithm for this algorithmic problem, known as a \emph{locally greedy} algorithm, studied in \cite{Company1978}, which was subsequently extended in \cite{Goundan2007}.  
The algorithm proceeds by fixing some arbitrary permutation of the multiset composed of $b_i$ $i$'s for each player $i$.  It then iteratively builds the allocation $\mathbf{S}$ where, on iteration $j$, it chooses $u \in \argmax_c \{ f(S_i \cup\{c\},\mathbf{S_{-i}}) - f(S_i,\mathbf{S_{-i}}) \}$ and adds $u$ to $S_i$, where $i$ is the $j$th element of the permutation. Regardless of the permutation selected, this algorithm is guaranteed to obtain a $2$-approximation to the optimal allocation subject to the given cardinality constraints \cite{Company1978,Goundan2007}.

\section{Counter examples when there are two agents}
\label{sec.examples}
The locally greedy algorithm \cite{Nemhauser1978} (see also \cite{Goundan2007})
is defined over an 
\emph{arbitrary} permutation of the agents allocation turns. 
In Section \ref{sec.main} we carefully construct such orderings in a manner
that induces strategyproofness for two players.
To motivate these algorithmic gymnastics, we now demonstrate that more natural
orderings fail to result in strategyproofness.

We begin by considering the
``dictatorship'' ordering, in which one player is first allocated
nodes up to his budget, and only then is the other player allocated nodes.
We will refer to the agents as
$A$ and $B$, and their utilities as $f_A$ and $f_B$ respectively; suppose that $A$ is the dictator. 
For the purposes of our example we will describe $f_A$ and $f_B$ in terms
of the following concrete (but simple) competitive influence spread process\footnote{This process is a simplification of the OR model \cite{springerlink:10.1007/978-3-642-17572-5_48}.} on an undirected network
$G = (V,E)$.

Suppose that each agent is given an initial seed set, say $S_A$ and $S_B$.
For agent $A$, each node in $S_A$ is given
a single chance to \emph{activate} each of its neighbors independently, which it does with probability $p=0.9$.  (Note that this activation process is not recursive; it affects
only the neighbors of $S_A$).
We then, independently, allow each node in seed set $S_B$ to attempt to activate each of
its neighbors,
resulting in a set of nodes activated by $B$.  To determine the final influence sets,
any node activated only by $A$ is influenced by $A$, any activated only by $B$ is influenced by $B$,
and any node activated by both will choose between the two agents uniformly at random.  The value
of $f_A(S_A, S_B)$ is the expected number of nodes influenced by $A$ at the end of this process,
and similarly for $f_B$.
One can easily show that an agent's influence is non-decreasing in its seed set, that
the sum of influences is submodular non-decreasing, and that the functions satisfy adverse competition.

Our network is as follows.  The graph consists of two components; one is the complete bipartite graph $K_{2,10}$, and the other is the star $K_{1,4}$.  Let $w_1$ and $w_2$ be the two nodes of degree $10$, and let $v$ be the center of the star.  We claim that the locally greedy algorithm paired with the dictatorship ordering is not strategyproof for this network.  Suppose each agent declares a budget of $1$; in this case, the algorithm will allocate $w_1$ to agent $A$, then it will allocate $v$ to agent $B$ (since $4p > 10(1-(1-p)^2)-10p$, which means that $v$ maximizes the marginal gain in \emph{social welfare}). This results in an expected influence of $10p = 9$ for agent $A$.  In the case where $A$ has a budget of $2$ (and $B$'s budget is still $1$), the greedy algorithm will allocate $w_1$ and $v$ to agent $A$ (for the same reason as before), and will give $w_2$ to agent $B$. In this case, the influence of agent $A$ becomes $4p + 10(p \cdot (1-p) + \frac{p^2}{2}) = 8.55 < 9$, so in particular his influence is not non-decreasing in his declaration.

The above construction can be modified to show that various other
 orderings for the locally greedy algorithm fail to result in strategyproof mechanisms.  
Appendix~\ref{sec.examples.appendix} provides the following examples:
\begin{enumerate}
\itemsep0em
\item The Round Robin ordering: the mechanism
alternates between the players when allocating a node. 
\item Always choosing the player having the smallest current unsatisfied
budget breaking ties in favor of player A. 
\item Taking a uniformly random choice over all orderings with the
required number of allocations to A and B.
\end{enumerate}
The last example is particularly relevant, since in Section~\ref{sec:3player}
we show that for the case of $k>2$ agents, in a setting
that assumes two additional restrictions called MeI and AgI (which
will be defined in Section~\ref{sec:3player}), taking a uniformly random permutation over the
allocation turns is a strategyproof algorithm and results in an
$\frac{e}{e-1}$ approximation to the optimal social welfare. 
In contrast, for the
case of $k=2$, and even with these additional restrictions (one can
verify that the influence model described above, used for our counterexample,
does satisfy both MeI and AgI, although the AgI condition is vacuous), 
the uniformly random mechanism is not strategyproof. \section{A Strategyproof Mechanism for Two Players}
\label{sec.main}

In this section we describe our mechanism for allocating nodes when there are two agents.  The case of $k > 2$ agents is handled in Section \ref{sec:3player}, under additional assumptions that are not necessary for the case $k=2$.
Our mechanism is based on the local greedy algorithm described in Section ~\ref{sec.preliminaries}.
We will focus on cases in which the allocations to the two agents need not be disjoint;
in Section~\ref{sec.disjoint} we extend our result to handle disjointness constraints when agents are ``anonymous.''

A nice property of the local greedy algorithm is that 
its worst-case approximation factor of $2$
holds even if we arbitrarily fix the order in which allocations are made to players $A$ and $B$.
This grants a degree of freedom that we will use to satisfy strategyproofness.
Given a particular pair of budgets $(a,b)$, we will randomize over possible orderings in which to allocate to the two agents, and then apply the greedy algorithm to whichever permutation we choose.  The key to the algorithm will be the manner in which we choose the distribution to randomize over, which will depend on the declared budgets and the influence functions $f_i$.  
As it turns out,
some of the more immediate ways of selecting an ordering 
lead to non-strategyproof mechanisms.  See Appendix~\ref{sec.examples} for a survey of na\"{i}ve orderings.  
Indeed, it is not even clear a priori that distributions exist that simultaneously monotonize the expected allocation for both players.  Our main technical contribution is a proof that such distributions do exist, and moreover can be explicitly constructed in polynomial time.

The idea behind our construction, at a high level, is as follows.  We will construct the distribution for use with budgets $(a,b)$ recursively.  Writing $t = a+b$, we first generate distributions for the case $t=1$ (which are trivial), followed by $t=2$, etc.  To construct the distribution for demands $(a,b)$, we consider the following thought experiment.  We will choose an ordering in one of two ways.  Either we choose a permutation according to the distribution for budget pair $(a-1,b)$ and then append a final allocation to $A$, or else choose a permutation according to the distribution for budget pair $(a,b-1)$ and append an allocation to player $B$.  If we choose the former option with some probability $\alpha$, and the latter with probability $1-\alpha$, this defines a probability distribution for budget pair $(a,b)$.  

What we will show is that, assuming our distributions are constructed to adhere to certain invariants, we can choose this $\alpha$ such that the resulting randomized algorithm (i.e.\ the greedy algorithm applied to permutations drawn from the constructed distributions) will be monotone.  That is, the expected influence of player $A$ under the distribution for $(a,b)$ is at least that of the distribution for $(a-1,b)$, and similarly for player $B$.  The existence of such an $\alpha$ is not guaranteed in general; we will need to prove that our constructed distributions satisfy an additional ``cross-monotonicity'' property in order to guarantee that such an $\alpha$ exists.

One problem with the above technique is that it does not bound the size of the support of the distributions.
In general there will be exponentially many possible permutations to randomize over, leading to exponential computational complexity to compute each $\alpha$.  One might attempt to overcome such issues by sampling to estimate the required probabilities,
but this introduces the possibility of non-monotonicities due to sampling error, which we would like to avoid.  We demonstrate that each distribution we construct can be ``pruned'' so that its support contains at most three permutations, while still retaining its monotonicity properties.  In this way, we guarantee that our recursive process requires only polynomially many queries (to the influence functions) in order to choose a permutation.

\subsection{The Allocation Algorithm}
Our algorithm will proceed by choosing a distribution over orders in which nodes are allocated to the two players.  This will be stored in a matrix $M$, where $M[a,b]$ contains a distribution over sequences $(y_1, \dotsc, y_t) \in \{A,B\}^{a+b}$, containing $a$ `A's and $b$ `B's.  We then choose a sequence from distribution $M[a,b]$ and greedily construct a final allocation with respect to that ordering.  We begin by describing the manner in which the allocation is made, given the distribution over orderings.  The algorithm is given as Algorithm \ref{alg:mech2}.  
\begin{algorithm}[ht]
\KwIn{Ground set $U= \{e_1,\ldots,\ldots,e_n\}$, budgets $a,b$ for players $A$ and $B$, respectively}
\KwOut{An allocation $I_A,I_B \subseteq U$ for the two players}
\BlankLine
\tcc{Build permutation table.}
$M \leftarrow ConstructDistributions(a,b)$ \;
\tcc{$M[a,b]$ will be a distribution over sequences $(y_1, \dotsc, y_{a+b}) \in \{A,B\}^{a+b}$}
Choose $(y_1, \dotsc, y_{a+b})$ from distribution $M[a,b]$\;
\For{$i = 1 \dotsc a+b$}{
  \eIf {$y_i = 'A'$}{ 
    $u \leftarrow argmax_{c \in U} \{ f(I_A \cup \{c\}, I_B) - f(I_A, I_B)\}$ \;
    $I_A \leftarrow I_A \cup \{u\}$ \;
  } {
    $u \leftarrow argmax_{c \in U} \{ f(I_A, I_B \cup \{c\}) - f(I_A, I_B)\}$ \;
    $I_B \leftarrow I_B \cup \{u\}$ \;
  }
}
\caption{Allocation Mechanism}
\label{alg:mech2}
\end{algorithm} 
An important property of the allocation algorithm that we will require for our analysis is that, given a sequence drawn from distribution $M[a,b]$, the allocation is chosen myopically.  That is, items are chosen for the players in the order dictated by the given sequence, independent of subsequent allocations.  We will use this property to construct the distribution $M[a,b]$, which will be tailored to the specific algorithm to ensure strategyproofness.  We note that this technique could be applied to \emph{any} allocation algorithm with this property; we will make use of this observation in Section \ref{sec.disjoint}.

Recall that the approximation guarantee for the greedy allocation does not depend on the order of assignment implemented in lines 3-11, so that the allocation returned by the algorithm will be a $2$-approximation to the optimal total influence regardless of the permutation chosen on line 2.  It remains only to demonstrate that we can construct our distributions in such a way that the expected payoff to each player is monotone increasing in his bid.
\subsection{Constructing matrix $M$}
We describe the procedure $ConstructDistributions$, used in 
Algorithm \ref{alg:mech2}, to generate distributions over orderings of assignments to players $A$ and $B$.  
We will build table $M[\cdot,\cdot]$ recursively, where $M[a,b]$ describes the distribution corresponding to 
budgets $a$ and $b$.  Our procedure will terminate when the required entry has been constructed.

We think of $M[a,b]$ as a distribution over sequences of the form $(y_1, \dotsc, y_{a+b})$, where 
$y_i \in \{A,B\}$.  For any given sequence, the corresponding allocation is determined
since the greedy algorithm applied in Algorithm \ref{alg:mech2} is deterministic.  We can therefore also think of
$M[a,b]$ as a distribution over allocations, and in what follows we will refer to ``allocations drawn from $M[a,b]$'' without
further comment.

Note that $M[0,b]$ must assign probability $1$ to the sequence $(B, B, \dotsc, B)$,
and similarly $M[a,0]$ assigns probability $1$ to $(A, A, \dotsc, A)$.  We will construct the remaining entries 
of the table $M[a,b]$ in increasing order of $a+b$.

Before describing the recursive procedure for filling the table, we provide some notation.
Given $M$, we will write $w^A(a,b)$ for the expected value of agent $A$ under the distribution of
allocations returned by $M[a,b]$.  Similarly, $w^B(a,b)$ will be the expected value of agent $B$, and
$w(a,b) = w^A(a,b) + w^B(a,b)$ is the expected total welfare.  For notational convenience, 
set $w^A(a,b) = w^B(a,b) = 0$ if $a < 0$ or $b < 0$.

We will construct $M$ so that the following invariants hold for all $a > 0$ and $b > 0$:
\begin{enumerate}
\item $w^A(a,b) \geq w^A(a-1,b)$.
\item $w^B(a,b) \geq w^B(a,b-1)$.
\item $w^A(a,b) \geq w^A(a-1,b+1)$.
\item $w^B(a,b) \geq w^B(a+1,b-1)$.
\item The support of $M[a,b]$ contains at most $3$ sequences.
\end{enumerate}
The first two desiderata capture the monotonicity properties we require of our algorithm.  Note that if $M$ satisfies these properties, then Algorithm \ref{alg:mech2} will be monotone and hence strategyproof.  The subsequent two invariants are cross-monotonicity properties, which will help us in the iterative construction of further entries of $M$. The final property limits the complexity of constructing and sampling from $M[a,b]$, implying that Algorithm \ref{alg:mech2} runs in polynomial time.

We now describe the way in which we construct distribution $M[a,b]$, given distributions $M[a',b']$ for all $a'+b' < a+b$.  We consider two distributions: the first selects a sequence according to $M[a-1,b]$ and appends an 'A', and the second selects a sequence according to $M[a,b-1]$ and appends a 'B'.  Call these two distributions $D_1$ and $D_2$, respectively.  What we would like to do is find some $\alpha$, $0 \leq \alpha \leq 1$, such that if we choose from distribution $D_1$ with probability $\alpha$ and distribution $D_2$ with probability $1 - \alpha$, then the resulting combined distribution (for $M[a,b]$) will satisfy $w^A(a,b) \geq w^A(a-1,b)$ and $w^B(a,b) \geq w^B(a,b-1)$.  Of course, this combined distribution may have support of size up to $6$ ($3$ from $D_1$ and $3$ from $D_2$) but we will show that it can be pruned to a distribution with the same expected influence for agents $A$ and $B$, with at most $3$ permutations in its support.

Our main technical result, Theorem \ref{lem:monotone2}, demonstrates that an appropriate value of $\alpha$, as described in the process sketched above, is guaranteed to exist and can be found efficiently.

\begin{theorem}
\label{lem:monotone2}
It is possible to construct table $M$ in such a way that the following
properties hold for all $a+b \geq 1$:
\begin{enumerate}
\item $w^A(a,b) \geq w^A(a-1,b)$
\item $w^B(a,b) \geq w^B(a,b-1)$.
\item $w^A(a,b) \geq w^A(a-1,b+1)$
\item $w^B(a,b) \geq w^B(a+1,b-1)$.
\end{enumerate}
Furthermore, the entries of $M$ can be computed in polynomial time.
\end{theorem}

\comment{
Notice that condition 3 in Lemma \ref{lem:monotone2} implies that
player $B$'s valuation is monotone increasing with his bid:
\begin{align}
w^A(a,b-1)  & \geq w^A(a,b) - \Delta^{\oplus B}(a,b) \nonumber \\
& = w^A(a,b) - [ w(a,b) - w(a,b-1) ] \nonumber \\
& = w^A(a,b) - \left[ \left(w^A(a,b) + w^B(a,b) \right) - \right. \nonumber \\
& \quad\quad\quad\quad\quad\quad - \left. \left(w^A(a,b-1) + w^B(a,b-1)\right) \right] \nonumber \\
& = w^A(a,b-1) + w^B(a,b-1) - w^B(a,b) \nonumber \\
& \Rightarrow  w^B(a,b) \geq w^B(a,b-1)
\end{align}}
\begin{proof}
We will proceed by induction on $t = a+b$.  The result is trivial for $t = 1$, so consider $t > 1$.

Our proof will be geometric.  We will associate with each entry of $M$, say $M[a,b]$, the point $(w^A(a,b), w^B(a,b))$ in $\R^2$.  Our approach will be to use the entries of $M$ with total budget less than $t$ to construct a certain convex region of potential points for each $M[a,b]$ with $a+b = t$.  We will then argue that there is a way to select a point from each convex region satisfying the 
four required properties.

We introduce some helpful notation.  For a point $p \in \R^2$, we let $p_A=p(1)$ denote the first coordinate of $p$, whereas the second coordinate of $p$ will be given by $p_B=p(2)$. 
Note that if $p$ is the point associated with $M[a,b]$, then $p_A = w^A(a,b)$.
Given a pair of points $p,q \in \R^2$, we will write $[p,q]$ for the line segment with endpoints $p$ and $q$.
Lastly, for two points $p, p' \in \R^2$, we say that $p \leq p'$ if $p_A \leq p_A'$ and $p_{B} \leq p_{B}'$.


The following simple geometric claim will be instrumental to our proof:
\begin{claim}
\label{cl:line_monotonicity}
Let $(L^{(1)}, \dotsc, L^{(t)})$ be a sequence of line segments in $R^2$, with $L^{(i)} = [p^{(i)}, q^{(i)}]$.  Suppose further that $p^{(i)} \leq q^{(j)}$ for all $i < j$.  Then there is a sequence of points $(r^{(1)}, \dotsc, r^{(t)})$, with $r^{(i)} \in L^{(i)}$, such that $r^{(i)} \leq r^{(j)}$ for all $i \leq j$.
\end{claim}

\begin{proof}
  By induction on $t$. The claim holds trivially for $t=1$.
  Suppose $t>1$ and the claim holds for all $t'<t$. Choose $r^{(1)}=p^{(1)}$. Also, for each $i >1$, define the line segment $\tilde L^{i}=\{z \in L^{(i)} : z \geq p^{(1)}\}$. Observe that (1) $\tilde L^{(i)} \subseteq L^{(i)}$, and (2) $q^{(i)} \in \tilde L^{(i)}$ which implies
$\tilde L^{(i)} \neq \emptyset$. Moreover $q^{(i)}$ is an endpoint of $\tilde L^{(i)}$ (by the assumption that $q^{(i)} \geq p^{(1)})$.
That is, $\tilde L^{(i)} = [\tilde p^{(i)},q^{(i)}]$ for some $\tilde p^{(i)}$. Note then that $\tilde p^{(i)}_A = \max \{p^{(i)}_A,p^{(1)}_A\}$ and $\tilde p^{(i)}_B = \max \{p^{(i)}_B,p^{(1)}_B\}$.

We now show that $\tilde p^{(i)} \leq q^{(j)}$ for all $2 \leq i < j$
by considering the $A$ and $B$ coordinates separately. Consider 
the $A$ coordinate for each of 
the two possible cases:
(1) $\tilde p^{(i)}_A = p^{(i)}_A \leq q^{(j)}_A$ by the assumption of the claim;
(2) $\tilde p^{(i)}_A = p^{(1)}_A \leq q^{(j)}_A$ by the definition of $\tilde L^{(i)}$. The proof for the $B$ coordinate is identical.    

We can now apply the inductive hypothesis on the sequence of sub-intervals $(\tilde L^{(2)},\ldots,\tilde L^{(t)})$ to obtain  a sequence of points 
$(r^{(2)}, \ldots, r^{(t)})$  satisfying $r^{(i)} \in \tilde L^{(i)}$ for all $i > 1$ and $r^{(i)} \leq r^{(j)}$ for all $1 < i < j$. Since $r^{(i)} \in \tilde L^{(i)}$, we also have $r^{(1)} \leq r^{(i)}$ for all $i > 1$. Thus, 
$(r^{(1)}, \ldots , r^{(t)})$ satisfies the required conditions for the claim. \qed
\end{proof}

Notice that the proof of Claim \ref{cl:line_monotonicity} is constructive, and the claimed sequence of points can be found in polynomial time given the sequence of line segments.  Our approach will be to argue that there is a sequence of line segments, one for each $a,b$ with $a+b=t$, satisfying the conditions of Claim \ref{cl:line_monotonicity}; we will then invoke the claim to select the point corresponding to each entry $M[a,b]$ with $a+b=t$. 

Recall that our construction of the DP table $M$ is inductive: in order to construct a distribution over allocation sequences for $M[a,b]$, such that $a+b=t$, we will rely on the two distributions that correspond to the entries $M[a-1,b]$ and $M[a,b-1]$. Let $v^{+A}(a-1,b)$ be defined as the 
pair $(w^A,w^B)$ of expected payoffs that results from sampling an allocation sequence from the entry $M[a-1,b]$ and appending an $`A'$ to it. We symmetrically define $v^{+B}(a,b-1)$ for the pair of expected payoffs resulting from
sampling a distribution from entry $M[a,b-1]$ and appending a $`B'$. 

For each $(a,b)$ with $a+b=t$, we define the line segment $L_c(a,b) = (v^{+A}(a-1,b), v^{+B}(a,b-1))$.  Note that this is the set of points (i.e., distributions) that \emph{can} be implemented by randomizing between the above two ``appending policies''.

Next, we define the following two half-spaces: $F_A(a,b) = \{ z \in \R^2: z_A \geq w^A(a-1,b) \}, F_B(a,b) =  \{ z \in \R^2: z_B \geq w^B(a,b-1) \}$.
We then let $L(a,b)= L_c(a,b) \cap F_A(a,b) \cap F_B(a,b)$.  Note that $L(a,b)$ is the set of points in $L_c(a,b)$ that satisfy monotonicity conditions 1 and 2 in our theorem statement.  We now show that $L(a,b)$ is non-empty; this will follow by induction from the fact that the entries of $M$ with total budget $t-1$ satisfy the four conditions conditions of the theorem.

\begin{lemma}
\label{lem:non-empty_inter}
  $L(a,b) \neq \emptyset$
\end{lemma}
\begin{proof}
We begin with the following notation and observations: 
\begin{enumerate}
\item Let $\ell_1= (x_1,y_1) = (w^{A}(a,b-1),w^B(a,b-1))$, and $\ell_2= (x_2,y_2) = (w^A(a-1,b),w^B(a-1,b))$.  That is, $\ell_1$ and $\ell_2$ are the points associated with $M[a,b-1]$ and $M[a-1,b]$. By the inductive hypothesis at $a+b-1 = t-1$, 
cross monotonicity implies $x_1 \geq x_2$ and $y_1 \leq y_2$.  
\item Let $p_1 = w^{A}(a,b-1) - v^{+B}(a,b-1)_A$, and  $q_1=v^{+B}(a,b-1)_B - w^{B}(a,b-1)$.  That is, $p_1$ (resp. $q_1$) is 
the loss in welfare to player $A$ (resp. the gain in welfare to Player $B$) 
when we add a node to $B$'s allocation starting from an allocation drawn from $M[a,b-1]$.
Note that $p_1 \leq q_1$ due to the monotonicity of the social welfare function.
\item Similarly, $p_2 = v^{+A}(a-1,b)_A  - w^{A}(a-1,b), q_2= w^{B}(a-1,b) - v^{+A}(a-1,b)_B$.  That is, $p_2$ (resp. $q_2$) is the gain in $A$'s welfare (resp. loss in $B$'s welfare) upon adding another node to $A's$ allocation starting from an allocation drawn from $M[a-1,b]$. It 
follows again that $p_2 \geq q_2$.
\end{enumerate}

Given this notation, Lemma \ref{lem:non-empty_inter} is implied by the following claim applied to $\ell_1, \ell_2, p_1, q_1, p_2, q_2$ as defined above. 

\begin{claim}
\label{cl:non_negative}
  Consider two points $\ell_1=(x_1,y_1), \ell_2=(x_2,y_2) \in \mathbb{R}^2$ such that $x_1 \geq x_2$ and $y_1 \leq y_2$.
  Let $e_1=(x_1 - p_1, y_1+q_1), e_2=(x_2+p_2,y_2-q_2) \in \mathbb{R}^2$ such that $0 \leq p_1 \leq q_1$ and $p_2 \geq q_2 \geq 0$. Then there exists a sub-interval of the line-segment $[e_1,e_2]$ that is contained in the region $D_{+}=\{(x,y) \in \mathbb{R}_+^2: x \geq x_1 \text{ and } y \geq y_2 \}$.
\end{claim}
\begin{proof}
By translating the line $[\ell_1,\ell_2]$ (maintaining the same slope), we can assume without loss of generality 
that $x_1=0, y_2=0$ and 
$x_2,y_1 \geq 0$, so that $D_{+}=\R^2_{\geq 0}$.

First,  if either $e_1$ or $e_2$ are in $\mathbb{R}^2_{\geq 0}$, the claim is trivial, so we assume that $e_1,e_2 \notin \mathbb{R}_{\geq 0}^2$, implying that $x < p_1$ and $y < q_2$.

Consider a point $e'$ on the line segment $[e_1,e_2]$, given by a convex combination $e'=\alpha \cdot e_1 + (1-\alpha) \cdot e_2$, for $\alpha \in [0,1]$. Having $e' \in \mathbb{R}^2_{\geq 0}$ is equivalent to the following two conditions:
\begin{align}
\alpha (x-p_1) + (1-\alpha) p_2  \geq 0 \label{ineq:1} \\
 \alpha q_1 + (1-\alpha) (y - q_2) \geq 0 \label{ineq:2}
\end{align}
Using our assumption that $x < p_1$, the first inequality gives
\[ \alpha \leq \frac{p_2}{p_1+p_2 -x }\]
whereas the second inequality gives (using our assumption that $y < q_2$):
\[ \alpha \geq \frac{q_2 - y}{q_1+q_2 - y}\]

Let $u= (q_2 - y)/(q_1+q_2 - y),  v = p_2/(p_1+p_2 -x )$.
Note that $u,v \in [0,1]$. Proving that $u \leq v$ will conclude the claim 
as that determines the range $u \leq \alpha \leq v$ defining the required sub-interval.
First, we have that
\begin{align*}
  u \leq \frac{p_2 - y}{q_1+p_2 - y} \leq \frac{p_2}{q_1+p_2}
\end{align*}
where the first inequality follows from our assumption that $p_2 \geq q_2$.
Second, our assumption that $p_1 \leq q_1$ implies that
\begin{align*}
  v \geq \frac{p_2}{q_1 + p_2 -x} \geq \frac{p_2}{q_1+p_2}
\end{align*}
where the last inequality follows again by our assumption that $x < p_1 (\leq q_1)$.
This shows that $u \leq v$, thereby concluding the proof.\qed
\end{proof}
\end{proof}

Write $p(a,b)$ for the endpoint of $L(a,b)$ that is closer to $v^{+A}(a-1,b)$, and $q(a,b)$ for the endpoint closer to $v^{+B}(a,b-1)$.
Consider the sequence of line segments $( L(t,0), L(t-1,1), \dotsc, L(0,t) )$ (i.e., the intervals corresponding to the $t$'th diagonal). To complete the 
proof of Lemma \ref{lem:monotone2}, it remains to show how to pick solutions from each of these intervals in a way that maintains the two cross-monotonicity constraints. 
 
It's enough to show that 
\[ p(t-i,i)_A \geq q(t-j,j)_A\]
and 
\[p(t-i,i)_B \leq q(t-j,j)_B\] 
for all $i < j$, since then Claim~\ref{cl:line_monotonicity} tells us how to choose a point from each line segment such that all of our cross-monotonicity properties are satisfied.  Notice that the first inequality goes in the opposite direction than would be implied by $p(t-i,i) \leq q(t-j,j)$. However, we can still apply Claim~\ref{cl:line_monotonicity} due to symmetry and the independence of the two axes. That is, we can reflect the line segments
in the $A$ axis so that $(p,q)$ becomes $(-p,q)$, apply Claim \ref{cl:line_monotonicity} to
obtain the points on each line, and then reflect back in the $A$ axis to 
obtain the required points.  

Let us prove the inequality on the first entry ($A$'s utility). Note that $p(t-i, i)_A \geq w^A(t-i-1,i)$ from the definition of $v^{+A}(a-1,b)$ and the intersecting half-space $F_A(a,b)$.  Furthermore, for all $j > i$, we have $q(t-j,j)_A \leq w^A(t-j, j-1)$. Indeed, if $q(t-j,j)=v^{+B}(t-j,j-1)$, then the 
inequality follows from the definition of $v^{+B}(a,b-1)$ and
adverse competition. Otherwise ($q \neq v^{+B}(t-j,j-1)$), the manner of the selection of $q(t-j,j)$ implies that $v^{+B}(t-j,j-1)_A < w^{A}(t-j-1,j)$, and so by the intersection with $F_A(a,b)$, we know that $q_(t-j,j)=w^{A}(t-j-1,j) \leq w^{A}(t-j,j-1)$ (by cross-monotonicity for $t-1$). Since $j > i$, we then have
\[ p(t-i,i)_A \geq w^A(t-i-1,i) \geq w^A(t-j, j-1) \geq q(t-j,j)_A \]
where the second inequality follows from cross-monotonicity 
(applied $j-i$ times) for budgets summing to $t-1$.

The argument for the second coordinate is identical. 
\comment{

 Here it is for completeness.  Note that $A(k-i, i).y \leq w^B(k-i-1,i)$ from the definition of $M^{+A}(a-1,b)$ with $(a,b) = (k-i,i)$.  Furthermore, for all $j > i$, we have $B(k-j,j).y \geq w^B(k-j, j-1)$ from the definition of $M^{+B}(a,b-1)$ and the intersecting halfspace $\{ A : A.y \geq w^B(a,b-1) \}$ with $(a,b) = (k-j,j)$.  Since $j > i$, we then have
\[ A(k-i,i).y \leq w^B(k-i-1,i) \leq w^B(k-j, j-1) \leq B(k-j,j).y \]
where we used cross-monotonicity for budgets summing to $k-1$.
}
We can now apply Claim \ref{cl:line_monotonicity} to choose points from each of the line segments that respect cross-monotonicity.  Those will be the points we use to populate the $M$ matrix on the diagonal corresponding to total budget $t$.
We have thus shown that table $M$ can be filled with distributions that
satisfy conditions 1-4 of Theorem \ref{lem:monotone2}.   

It remains to discuss the
complexity of computing the entries of $M$.  To this point we have not 
bounded the size of our distributions' supports.  We will
modify the argument to show that the number of permutations
required for each table entry $M[a,b]$ can be limited to only three,
by induction on $t$.

Consider the distribution constructed for $M[a,b]$.  The support of this distribution
has size at most $6$: the three permutations in the support of $M[a-1,b]$ with 
$A$ appended, plus the three permutations in the support of $M[a,b-1]$ with $B$
appended.  Each of these six permutations implies an allocation, say $(S_1, T_1), \dotsc, (S_6, T_6)$.
For each allocation, we consider the two-dimensional
point $(f_A(S_i,T_i), f_B(S_i,T_i))$ representing the welfare to A and B for the given
allocation.  We can interpret our construction of $M[a,b]$ as implementing a 
point $(w^A(a,b), w^B(a,b))$ with certain properties, such that this point lies in the convex hull of the six points
$(f_A(S_1,T_1), f_B(S_1,T_1)), \dotsc, (f_A(S_6,T_6), f_B(S_6,T_6))$.

We now use the following well-known theorem \cite{citeulike:340551}: 
\begin{theorem}[Carathéodory]
Given a set $V \subset \mathbb{R}^n$ and a point $p \in Conv V$ --- the convex hull of $V$, there exists a subset $A \subset V$ such that $|A| \leq n + 1$ and $p \in Conv A$.
\end{theorem}
It must therefore be that our point $(w^A(a,b), w^B(a,b))$ lies in the convex hull of at most three of the points $(f_A(S_1,T_1), \\f_B(S_1,T_1)), \dotsc, (f_A(S_6,T_6), f_B(S_6,T_6))$.   It follows that there exists a distribution with a support that consists of three of the six permutations corresponding to $(a,b)$. Finding this distribution can be done in constant time by considering $\binom{6}{3}$ sets of three allocations.\footnote{Note that all quantities in this geometric problem are rational numbers, which are constructed via the sequence of operations described in the proof above and therefore have polynomial bit complexity. We can therefore solve the convex hull tasks described in this operation in polynomial time.}
We can therefore construct $M[a,b]$ as a distribution over at most $3$ permutations, concluding the proof of Theorem~\ref{lem:monotone2}.
\end{proof} \qed

The proof of Theorem \ref{lem:monotone2} is constructive: it implies a recursive method for constructing
the table $M$ of distributions.  That is, the procedure \emph{ConstructDistributions} from Algorithm \ref{alg:mech2}
(with input $(a,b)$) will procede by filling table $M$ in increasing order of $t$, up to $a+b$, by choosing the value
of $\alpha$ for each table entry as in the proof of Theorem \ref{lem:monotone2}, then storing the implied
distribution over three permutations.  Note that we can explicitly store the allocations corresponding to the
permutations in the table.
We conclude, given this implementation of  \emph{ConstructDistributions}, that Algorithm \ref{alg:mech2}
is a polytime strategyproof $2$-approximation to the $2$-player influence maximization problem.

 
\section{A Strategyproof Mechanism for Three or More Players}
\label{sec:3player}
To construct a strategyproof mechanism for $k > 2$ players, 
we will impose additional restrictions on the influence functions $f_1, \dotsc, f_k$. 
These additional assumptions are satisfied by
many models of influence spread considered in the literature, as we discuss below.
We show that, under these assumptions, there is a natural mechanism that is strategyproof when there are at 
least three players.  In fact, it turns
out that having three or more players in such a setting allows for a
much simpler mechanism than the mechanism for the case of only two players\footnote{At this point, the reader may wonder if the two player case can be reduced to the case $k > 2$ by adding dummy agents with budget $0$.  This does not work because strategyproofness is defined over the space of \emph{all possible} agent bids, so 
we cannot restrict our attention only to profiles in which some players bid 0.  
Our examples in Appendix~\ref{sec.examples} show that this is not just a nuance of the proof but rather an intrinsic obstacle to using the uniform distribution.}.
\paragraph{Assumption 1: Mechanism Indifference}
We will assume that $f(\mb{S}) = f(\mb{S}')$ whenever the sets $\bigcup_i S_i$ and $\bigcup_i S_i'$
are equal 
\footnote{We again postpone discussion of multiset allocations
until Appendix~\ref{sec.models}}. That is, social welfare does not depend on the manner in which allocated items are partitioned
between the agents.  We will call this the \emph{Mechanism Indifference} (MeI) assumption.

\medskip

If assumption 1 holds, then we can imagine a greedy algorithm that chooses which items to add to
the set
$\bigcup_i S_i$ 
one at a time to greedily maximize the social welfare.  By assumption 1,
the welfare does not depend on how these items are divided among the players.  This
greedy algorithm generates a certain social welfare whenever the sum of budgets is $t$; 
write $w(t)$ for this welfare.  Note that $w(0), w(1), \dotsc$ is a 
concave non-decreasing 
sequence.



\paragraph{Assumption 2: Agent Indifference}
We will assume that $f_i(S_i, \mb{S}_{-i}) = f_i(S_i, \mb{S}'_{-i})$ whenever sets
$\bigcup_{j \neq i} S_j$ and $\bigcup_{j \neq i} S_j'$ are equal.  That is, each agent's utility depends on the
set of items allocated to the other players, but not on how the items are partitioned among those
players.  We will call this the \emph{Agent Indifference} (AgI)
assumption. Notice that in the two-players case, this assumption is
essentially vacuous.

\medskip


We note that the 
models for competitive influence spread
proposed by Carnes et al. \cite{Carnes:2007:MIC:1282100.1282167} 
and Bharathi et al. \cite{DBLP:conf/wine/2007/BharathiKS07} are based
on a cascade model of influence spread, and satisfy 
both the MeI and AgI assumptions. Similarly, if we restrict the
OR model in \cite{springerlink:10.1007/978-3-642-17572-5_48} so that the
underlying spread process is a cascade (and not a threshold) process 
and agents are anonymous (a restriction that will be defined in
Section~\ref{sec.disjoint}), as assumed in the Carnes et al models,
then this special case of the OR model also satisfies MeI and AgI.


\subsection{The uniform random greedy mechanism}
Consider Algorithm \ref{alg:rand.mech}, which we refer to as the 
uniform random greedy mechanism.  This mechanism proceeds by first greedily selecting which subset of the  ground set elements to allocate.  It then chooses an ordering of the players' bids uniformly at random from the set of all possible orderings, then assigns the selected elements to the players in this randomly chosen order. Note that this always results in a disjoint allocation. 
\begin{algorithm}
\KwIn{Ground set $U= \{e_1,\ldots,e_m\}$, budget profile $\mb{b}$}
\KwOut{An allocation profile $\mb{S}$}
\BlankLine
Initialize: $S_i \leftarrow \emptyset, i \leftarrow 0, j \leftarrow 0, I \leftarrow \emptyset, t \leftarrow \sum_i b_i$\;
\tcc{Choose elements to assign.}
\While{$i < t$}{ \label{alg:rand.mech:while.begin}
  $u_i \leftarrow argmax_{c \in U} \{f(I \cup \{c\}) - f(I)\}$ \;
  $I \leftarrow I \cup \{u_i\}$\ ; 
  $i \leftarrow i + 1$ \; 
} \label{alg:rand.mech:while.end}
\tcc{Partition elements of $I$.}
$\Gamma \leftarrow \{ \beta : [t] \to [k] \text{ s.t. } |\beta^{-1}(i)| = b_i \text{ for all } i \}$ \; 
Choose $\beta \in \Gamma$ uniformly at random \; 
\While{$j < t$}{
  $S_{\beta(j)} \leftarrow S_{\beta(j)} \cup \{u_j\}$ \; 
  $j \leftarrow j + 1$ \;
}
\caption{Uniform Random Greedy Mechanism}
\label{alg:rand.mech}
\end{algorithm} 
The MeI assumption implies that the random greedy mechanism obtains a constant factor approximation to the optimal social welfare.  We now claim that, under the MeI and AgI assumptions, this mechanism is strategyproof as long as there are at least $3$ players.

\begin{theorem}
\label{thm.random}
If there are $k \geq 3$ players and the AgI and MeI assumptions hold, then Algorithm
\ref{alg:rand.mech} is a strategyproof mechanism. Furthermore,
Algorithm~\ref{alg:rand.mech} approximates the social welfare to
within a factor of $\frac{e}{e-1}$ from the optimum.
\end{theorem}
\begin{proof}
As before, notice that lines
\ref{alg:rand.mech:while.begin}--\ref{alg:rand.mech:while.end} are an
implementation of the standard greedy algorithm for maximizing a
non-decreasing, submodular set-function subject to a uniform matroid
constraint, as described in \cite{Company1978,Goundan2007}, and hence
gives the specified approximation ratio.

Next, we show that Algorithm~\ref{alg:rand.mech} is strategyproof. Fix
bid profile $\mb{b}$ and let $t = \sum_i b_i$.  Let $I$ be the union
of all allocations made by Algorithm \ref{alg:rand.mech} on bid
profile $\mb{b}$; note that $I$ depends only on $t$.  Furthermore,
each agent $i$ will be allocated a uniformly random subset of $I$ of
size $b_i$.  Thus, the expected utility of agent $i$ can be expressed
as a function of $b_i$ and $t$.  We can therefore write $w^i(b, t)$
for the expected utility of agent $i$ when $b_i = b$ and $\sum_{j} b_j
= t$ (recall that we let $w(t)$ denote the total social welfare when $\sum_i b_i = t$).

We now claim that $w^i(b,t) = \frac{b}{t} w(t)$ for all $i$ and all $0 \leq b \leq t$.  Note that this implies the desired result, since if our claim is true then for all $i$ and all $0 \leq b \leq t$ we will have
\[ w^i(b,t) = \frac{b}{t} w(t) \leq \frac{b+1}{t+1} w(t+1) = w^i(b+1,t+1) \]
which implies the required monotonicity condition.  

It now remains to prove the claim.
The adverse competition assumption implies that $w^i(0,t) \leq w^i(0,0) = 0$ for all $i$ and $t$.  We next show that $w^i(1,t) = w^j(1,t)$ for all $i$, $j$, and $t \geq 1$.  If $t = 1$ then this follows from the MeI assumption.  So take $t \geq 2$ and pick any three agents $i$, $j$, and $\ell$.  
Then, by the AgI assumption, we have
\[w^i(1,t) = w(t) - w^{\ell}(t-1,t) = w^j(1,t).\]

We next show that $w^i(b,t) = w^i(1,t) + w^i(b-1,t)$ for all $i$, all $b \geq 2$, and all $t \geq b$.  Pick any three agents $i$, $j$, and $\ell$, any $b \geq 2$, and any $t \geq b$.  Then, by the AgI assumption,
\begin{align*} 
w^i(b,t) & = w(t) - w^{\ell}(t-b,t) \\
& = w(t) - [w(t) - w^i(b-1,t) - w^j(1,t)] \\
& = w^i(b-1,t) + w^j(1,t) \\
& = w^i(b-1,t) + w^i(1,t).
\end{align*}

It then follows by simple induction that $w^i(b,t) = bw^i(1,t)$ for all $1 \leq b \leq t$.  But now note that $w(t) = w^i(1,t) + w^j(t-1,t) = tw^i(1,t)$, and hence $w^i(1,t) = \frac{1}{t}w(t)$ and therefore $w^i(b,t) = \frac{b}{t} w(t)$ for all $0 \leq b \leq t$, as required.  
\end{proof}

Note that the proof of Theorem \ref{thm.random} makes crucial use of
the fact that there are at least three players.  Indeed, in Appendix \ref{sec.examples.appendix} we give an example satisfing the MeI and AgI assumptions for which the random greedy algorithm is not strategyproof for two players.

 \section{Disjoint Allocations}
\label{sec.disjoint}

One feature of the mechanism from Section \ref{sec.main} for $2$ players is that it does not necessarily return a profile of disjoint allocations.  That is, the algorithm may allocate a given element to both agents.  Disjointness is a natural property to require in many models of influence in social networks.  Note that the mechanism from Section \ref{sec:3player} always returns disjoint allocations.  

We now show how to modify the mechanism from Section \ref{sec.main} to ensure disjoint allocations.
Recall that our general strategy in the non-disjoint case was to 
use the locally greedy algorithm
and construct a strategyproof-inducing distribution over player orderings for that algorithm.
Our strategy for achieving disjointness will be to modify the underlying greedy algorithm so that it only returns disjoint allocations, then apply the same techniques as in Section \ref{sec.main} to convert this algorithm into a strategyproof mechanism.  As noted in Section \ref{sec.main}, our method can be applied to \emph{any} myopic allocation with a social welfare guarantee that does not depend on the chosen order of players.  It therefore suffices to find such a myopic allocation method that guarantees disjointness.


When the disjointness constraint is combined with demand restrictions, 
the set of valid allocations is not a matroid but rather an intersection of two matroids.  The locally greedy algorithm described in Section \ref{sec.preliminaries} is therefore not guaranteed to obtain a constant approximation for every ordering of the players.  
%
For example, 
suppose the ground set $U$ consists of two items, $1$ and $2$. Suppose player $A$ has values $1$ and $1+\epsilon$ for items $1$
and $2$, respectively (where $\epsilon > 0$ is arbitrarily small), and player $B$ has values $1$ and $N$ for items $1$ and $2$, respectively (where $N > 1$ is arbitrarily large).  When the demands of the two players are $1$, 
the locally greedy algorithm might allocate to either player first, but if it allocates to player $A$ first then it obtains the unbounded approximation ratio
$\frac{N+1}{2+\epsilon}$.  

The above problem stems from the asymmetry in the valuations of the two players.  To address this issue, we introduce a notion of player anonymity
that captures those circumstances in which these problems do not occur.

\begin{definition}
We say agents are \emph{anonymous} if 
their valuations are symmetric; that is,
 
$f_i(S_1, \ldots, S_k) = f_{\pi(i)}(S_{\pi(1)}, \ldots, S_{\pi(k)})$ for
all permutations $\pi$ and all agents $ 1 \leq i \leq k$. 
\end{definition}

If players are anonymous then the social welfare satisfies 
$f(S_1, \ldots, S_k) = f(S_{\pi(1)}, \ldots, S_{\pi(k)})$ 
for all permutations $\pi$.
We note that the influence
models proposed by Carnes et al. \cite{Carnes:2007:MIC:1282100.1282167} 
and Bharathi et al. \cite{DBLP:conf/wine/2007/BharathiKS07} satisfy this condition.
At the end of this section we will discuss the relationship between the anonymity
condition and the Agent Indifference and Mechanism Indifference conditions
from Section \ref{sec:3player}.

What we now show is that when the players are anonymous, our
order-independent locally greedy algorithm from Section \ref{sec.preliminaries} obtains a strategyproof mechanism with a $(k+1)$-approximation to the optimal social welfare, if the given permutation over orderings of the player allocations is sampled from a truthfulness-inducing distribution over permutations (e.g. the distributions we have obtained in the case of two players). Hence, this method provides a transformation to the disjoint allocations case, if one were to obtain a distribution over permutations for the non-disjoint case.

%

Algorithm~\ref{alg:mech3} is a simple modification to Algorithm~\ref{alg:mech2}, in which we explicitly enforce disjointness of the allocations.

\begin{algorithm}[ht]
\KwIn{Ground set $U= \{e_1,\ldots,\ldots,e_n\}$, demands $a,b$ for players $1,\ldots,k$, a valid permutation $\pi: \{1,2,\ldots,t\} \rightarrow \{1,2,\ldots ,k\}$ where $t=\sum_{i=1}^k b_i$}
\KwOut{An allocation $I_1, \ldots, I_k \subseteq U$ for the $k$ players}
\BlankLine
\For{$j = 1 \dotsc k$}{$I_j = \emptyset$} 
\For{$i = 1 \dotsc b_1 + \ldots + b_k$}{
    $u \leftarrow argmax_{c \in U-\bigcup I_j} \{w(I_{\pi(i)} \cup \{c\}, I_{\mathbf{-i}}) - w(I_{\pi(i)}, I_{\mathbf{-i}})\}$ \;
    $I_{\pi(i)} \leftarrow I_{\pi(i)} \cup \{u\}$ \;
}
\caption{Disjoint Locally Greedy algorithm}
\label{alg:mech3}
\end{algorithm}


\begin{theorem}
\label{thm:disjoint:appendix}
For any permutation (of player allocation turns) 
$\pi: \{1,2,\ldots,t\} \rightarrow \{1,2,\ldots ,k\}$ where $t = \sum_{i=1}^k b_i$, Algorithm~\ref{alg:mech3} obtains $(k+1)$-approximation to the optimal social welfare obtainable for disjoint allocation for identical players $1,\ldots, k$.
\end{theorem}
\begin{proof}
Let $\mathbf{O}=(O_1,\ldots,O_k)$ be an optimal allocation. Let $(I_1,\ldots, I_k)$
be the allocation obtained by running Algorithm~\ref{alg:mech3} for some permutation $\pi$.
Partition $\mathbf{O}$ as follows. For each player $i$, set $O_i^j=O_i \cap I_j$ for all $j \neq i$, and let $O_i^0=O_i - \bigcup_{j \neq i} I_j$. By submodularity,
\begin{equation}
\label{eq.disjoint1}
	w(O_1,\ldots, O_k) \leq w(O_1^0, \ldots, O_k^0) + \sum_{i=1}^{k-1}w(O_1^{(1+i) mod k},\ldots,O_k^{(k+i) mod k}).
\end{equation}
Due to anonymity and the fact that $O_i^j \subseteq I_j$, for all $j \in [k]$ we get
\begin{equation}
\label{eq.disjoint2}
  \sum_{i=1}^{k-1}w(O_1^{(1+i)mod k},\ldots,O_k^{(k+i)mod k}) \leq (k-1) \cdot w(I_1,\ldots, I_k).
\end{equation}

Next, we can apply the analysis performed for the original locally greedy algorithm, so as to obtain the relation specified in the following lemma
\begin{lemma}
\label{lem.disjoint.1}
$w(O_1^0,\ldots, O_k^0) \leq 2 \cdot w(I_1,\ldots,I_k).$
\end{lemma}
For readability, we prove the lemma at the end of this section.
Now, combining relations (\ref{eq.disjoint1}) and (\ref{eq.disjoint2}) 
and applying Lemma \ref{lem.disjoint.1}, we get:
\begin{equation}
	w(O_1,\ldots, O_k) \leq 2 \cdot w(I_1,\ldots,I_k) + (k-1) w(I_1, \ldots, I_k) = (k+1) \cdot w(I_1,\ldots,I_k)
\end{equation}
\end{proof}
\begin{corollary}
For $k = 2$ players, there exists a 3-approximate strategyproof mechanism that 
achieves disjoint allocations.
\end{corollary}
\begin{proof}
Observe that this revised version of the locally greedy algorithm is
order-independent. That is, we obtain the same
(constant) bound on its approximation ratio for any player ordering.  In particular, this means that we can apply the mechanism described in Section~\ref{sec.main} for obtaining a strategyproof 
solution without significantly degrading the approximation ratio of the greedy algorithm.
\end{proof}

We note that for all $k$ there is a natural greedy algorithm for this problem
(with disjointness) that obtains a $3$-approximation for any $k$.  Namely, the greedy algorithm that chooses
both the player and the allocation that maximizes the marginal utility
on each iteration \cite{Nemhauser1978}).  However, this algorithm imposes a particular ordering on the allocations and therefore does not allow the degree of freedom required by our truthful mechanism construction.

\begin{proof}[Proof of Lemma~\ref{lem.disjoint.1}]
We now adapt the analysis of the locally greedy algorithm
(e.g. \cite{Goundan2007}) in order to prove the bound in Lemma~\ref{lem.disjoint.1}. We
begin by introducing some additional notation. First, for $i \in [k]$,
let $O'_i = O_i^0 \setminus I_i$. For an item $e \in I_i$ ($i \in
[k]$), $\mathbf{S}^e=(S_1^e,\ldots,S_k^e)$ denotes the partial
solution of the algorithm at the time of item $e$'s addition. For an
allocation $(A_1,\ldots,A_k) \subseteq U^k$ and an item $e \in U$,
define the marginal gain to be $\rho_e^i(A_1,\ldots,A_k)=w(A_1,\ldots,A_i \cup \{e\},\ldots,A_k)-w(A_1,\ldots,\ldots,A_k)$. Lastly, for $i \in [k]$, we let $e_i = \arg\min_{e \in I_i}\rho_e^i(\mathbf{S}^{e})$; i.e. the minimal marginal increase to social welfare, obtained by the algorithm, when adding an element to $I_i$. The following lemma follows from the fact that the social welfare is a non-decreasing, submodular function:
\begin{lemma}[\cite{Company1978}]
\begin{equation*}
w(O_1^0,\ldots,O_k) \leq w(I_1,\ldots,I_k) + \sum_{i=1}^{k}\sum_{e \in O'_i}\rho_e^i(I_1,\ldots,I_k)
\end{equation*}
\end{lemma}

Now, by the greedy rule of Algorithm~\ref{alg:mech3}, we have:
\begin{equation}
  \label{eq.disjoint.2}
  \rho_{e_i}^i(S^e_i) \geq \rho_e^i(S^e_i), \text{for all $i \in [k]$, and for all $e \in U \setminus S_i^{e_i}$}
\end{equation}

Additionally, using the submodularity of $w(\cdot)$, and the fact that for all $i,j \in [k]$, $S^{e_i}_j \subseteq I_j$, we get:
\begin{align}
  \label{eq.disjoint.3}
  w(O_1^0,\ldots,O_k^0) &\leq w(I_1, \ldots, I_k) + \sum_{i=1}^{k}\sum_{e \in O'_i}\rho_e^i(I_1,\ldots,I_k) \nonumber \\
  &\leq w(I_1, \ldots, I_k) + \sum_{i=1}^{k}\sum_{e \in O'_i}\rho_e^i(\mathbf{S}^{e_i})
\end{align}
Now, using $(\ref{eq.disjoint.2})$, we further extend the above bound as follows
\begin{align}
  \label{eq.disjoint.4}
  w(O_1^0,\ldots,O_k^0) &\leq w(I_1, \ldots, I_k) + \sum_{i=1}^{k}\sum_{e \in O'_i}\rho_{e_i}^i(\mathbf{S^{e_i}}) \nonumber \\
  &= w(I_1, \ldots, I_k) + \sum_{i=1}b_i \cdot \rho_{e_i}^i(\mathbf{S^{e_i}}) \nonumber \\
  &\leq w(I_1,\ldots,I_k) +  w(I_1,\ldots,I_k) = 2 \cdot  w(I_1,\ldots,I_k)
\end{align}
where the last inequality follows from the definition of the elements $e_1,\ldots, e_k$.
\end{proof}

\subsection{Relation to Indifference Conditions}

In this section we explore the relationship between the anonymity condition required by Theorem \ref{thm:disjoint:appendix} and the mechanism and agent indifference conditions (MeI and AgI) used in Section \ref{sec:3player}. As we will show, these conditions are incomparable when there are only two players, but when there are three or more players the AgI and MeI conditions together are strictly stronger than the anonymity condition.  

Consider first the case of two players.  
To see that MeI does not imply anonymity, consider the following example with two objects $\{a,b\}$ and two players.  The functions $f_1$ and $f_2$ are given by $f_1(x,0) = f_2(0,x) = 2$ for any singleton $x$, $f_1(\{a,b\},0) = f_2(0,\{a_b\}) = 3$, and $f_1(x,y) = 1.6$, $f_2(x,y) = 1.4$ for $(x,y) = (a,b)$ or vice-versa.  One can verify that $f = f_1 + f_2$ is submodular and that adverse competition and mechanism indifference are satisfied, but it is not anonymous (since $f_1(x,y) \neq f_2(y,x)$ for singletons $x$ and $y$).

To see that anonymity does not imply MeI, consider the following example with two objects $\{a,b\}$ and two players.  We will have $f_1(x,0) = f_2(0,x) = 1$ for each singleton $x$, $f_1(\{a,b\},0) = f_2(0,\{a,b\}) = 2$, but $f_1(x,y) = f_2(x,y) = 3/4$ for $(x,y) = (a,b)$ or vice-versa.  This pair of functions exhibits adverse competition and its sum is submodular, but it does not satisfy MeI (since $f(a,b) \neq f(\{a,b\},0)$).

For $k \geq 3$ players, MeI and AgI together imply anonymity.

\begin{theorem}
If there are $k \geq 3$ agents and the AgI and MeI conditions hold, then the agents are anonymous.
\end{theorem}
\begin{proof}
We will assume $k=3$ for notational convenience; extending to $k>3$ is straightforward.  It is sufficient to show that $f_1(S,T,U) = f_2(T,S,U)$ for arbitrary sets $S,T,U$; symmetry with respect to all other permutations then follows by composing transpositions.

We first show that $f_1(S,\emptyset,U) = f_2(\emptyset,S,U)$ for all sets $S$ and $U$.  By MeI, $f(S,\emptyset,U) = f(\emptyset,S,U)$.  By AgI, $f_3(S,\emptyset,U) = f_3(\emptyset,S,U)$.  By our assumed normalization, $f_2(S,\emptyset,U) = f_1(\emptyset,S,U) = 0$.  Taking sums, we conclude that $f_1(S,\emptyset,U) = f_2(\emptyset,S,U)$.

By the same argument, $f_1(S,\emptyset,U \cup T) = f_2(\emptyset,S,U \cup T)$.  But then, by AgI, $f_1(S,T,U) = f_1(S,\emptyset,U+T) = f_2(\emptyset,S,U+T) = f_2(T,S,U)$, as required.
\end{proof}
Finally, we show that the MeI and AgI assumptions together are strictly stronger than anonymity for $k \geq 3$ players, as anonymity does not imply MeI.  Consider the following example with 3 objects $\{a_1, a_2, a_3\}$ and 3 players.  For any labeling of the singletons as $x$, $y$, $z$, define $f_1(x, y, z) = 7/24$, $f_1(\{x,y\},z,0) = f_1(\{x,y\},0,z) = 3/4$, $f_1(x, \{y,z\},0) = f_1(x,0,\{y,z\}) = 1/4$, and $f_1(\{x,y,z\},0,0) = 1$.  Define $f_2$ and $f_3$ symmetrically, so agents are anonymous.  Adverse competition is satisfied and the sum of these functions is submodular, but neither MeI nor AgI are satisfied.

 \section{Conclusions}
We have presented a general framework for mechanisms that allocate items given an underlying submodular process. Although we have explicitly referred to spread processes over social networks, we only require oracle access to the outcome values, and thus our methods apply to any similar settings which uphold the properties we have required from the processes. We build on natural greedy algorithms to construct efficient strategyproof mechanisms that guarantee constant approximations to the social welfare. 

An important question is how to extend our results to the more general case of 
$k > 2$ agents without the MeI and AgI assumptions. 
As discussed in Appendix~\ref{sec.3player}, it seems that a fundamentally new approach would be required to obtain an $O(1)$-approximate strategyproof mechanism for $k > 2$ players. 
Another natural and challenging extension would be
to assume that nodes have costs for being initially allocated and then
replace the cardinality constraint on each agent by a knapsack constraint.
To do so, the most direct approach would be to try to utilize the known
approximation for maximizing a non decreasing submodular function
subject to one \cite{S04} or multiple \cite{KST09} knapsack constraints.
These methods do not seem to readily lend themselves to the approach we
have been able to exploit in the case of cardinality constraints. We have 
also assumed a ``demand satisfaction'' condition. Without this condition, it is trivial to achieve a strategyproof $k$ approximation by allocating
all initial elements to the agent who can achieve the most utlility. Can we improve upon this trivial approximation factor without imposing the MeI and AgI 
assumptions. 


An important question is how to extend our results to the more general case of $k$ players.  As we discuss in 
Appendix \ref{sec.3player},
it seems that a fundamentally new approach would be required to obtain a polytime strategyproof mechanism for $k > 2$ players.

 
\paragraph{Acknowledgements} We would like to thank Yuval Filmus for many helpful comments and discussions. We are also indebted to Eyal Shani for pointing
out a significant problem in \cite{BorodinBLO13} (and \cite{BBLO12})  
of the proof that $L(a,b) \neq \emptyset$ 
(Lemma \ref{lem:non-empty_inter} in this paper).

{
\bibliographystyle{splncs}
\bibliography{sigproc}

\begin{thebibliography}{10}

\bibitem{G78}
Granovetter, M.:
\newblock Threshold models of collective behavior.
\newblock American Journal of Sociology (1978)  1420--1443

\bibitem{S78}
Schelling, T.:
\newblock Micromotives and macrobehavior.
\newblock Norton (1978)

\bibitem{BR87}
Brown, J.J., Reingen, P.H.:
\newblock Social ties and word of mouth referral behavior.
\newblock Journal of Consumer Research \textbf{14} (1987)  350--362

\bibitem{GLM01}
Goldenberg, J., Libai, B., Mulle, E.:
\newblock Talk of the network: A complex systems look at the underlying process
  of word-of-mouth.
\newblock Marketing Letters (2001)  221--223

\bibitem{CM07}
Centola, D., Macy, M.:
\newblock Complex contagions and the weakness of long ties.
\newblock American Journal of Sociology \textbf{113} (2007)  702--734

\bibitem{Kempe:2003}
Kempe, D., Kleinberg, J., Tardos, E.:
\newblock Maximizing the spread of influence through a social network.
\newblock In: KDD '03: Proceedings of the ninth ACM SIGKDD international
  conference on Knowledge discovery and data mining, New York, NY, USA, ACM
  (2003)  137--146

\bibitem{Kempe:2005}
Kempe, D., Kleinberg, J., Tardos, {\'E}.:
\newblock Influential nodes in a diffusion model for social networks.
\newblock In Caires, L., Italiano, G., Monteiro, L., Palamidessi, C., Yung, M.,
  eds.: Automata, Languages and Programming. Volume 3580 of Lecture Notes in
  Computer Science.
\newblock Springer Berlin / Heidelberg (2005)  1127--1138

\bibitem{DBLP:journals/siamcomp/MosselR10}
Mossel, E., Roch, S.:
\newblock Submodularity of influence in social networks: From local to global.
\newblock SIAM J. Comput. \textbf{39} (2010)  2176--2188

\bibitem{Nemhauser1978}
Nemhauser, G.L., Wolsey, L.A., Fisher, M.L.:
\newblock {An analysis of approximations for maximizing submodular set
  functions---II}.
\newblock Mathematical Programming \textbf{14} (1978)  265--294

\bibitem{Carnes:2007:MIC:1282100.1282167}
Carnes, T., Nagarajan, C., Wild, S.M., van Zuylen, A.:
\newblock Maximizing influence in a competitive social network: a follower's
  perspective.
\newblock In: Proceedings of the ninth international conference on Electronic
  commerce. ICEC '07, New York, NY, USA, ACM (2007)  351--360

\bibitem{DBLP:conf/wine/2007/BharathiKS07}
Bharathi, S., Kempe, D., Salek, M.:
\newblock Competitive influence maximization in social networks.
\newblock In Deng, X., Graham, F., eds.: Internet and Network Economics. Volume
  4858 of Lecture Notes in Computer Science.
\newblock Springer Berlin / Heidelberg (2007)  306--311

\bibitem{springerlink:10.1007/978-3-642-17572-5_48}
Borodin, A., Filmus, Y., Oren, J.:
\newblock Threshold models for competitive influence in social networks.
\newblock In Saberi, A., ed.: Internet and Network Economics. Volume 6484 of
  Lecture Notes in Computer Science.
\newblock Springer Berlin / Heidelberg (2010)  539--550

\bibitem{GoyalK12}
Goyal, S., Kearns, M.:
\newblock Competitive contagion in networks.
\newblock In: STOC. (2012)  759--774

\bibitem{FMMM09}
Feldman, J., Mehta, A., Mirrokni, V., Muthukrishnan, S.:
\newblock Online stochastic matching: Beating 1-1/e.
\newblock FOCS \textbf{0} (2009)  117--126

\bibitem{FHKMS10}
Feldman, J., Henzinger, M., Korula, N., Mirrokni, V.S., Stein, C.:
\newblock Online stochastic packing applied to display ad allocation.
\newblock In: Proceedings of the 18th annual European conference on Algorithms:
  Part I. ESA'10, Berlin, Heidelberg, Springer-Verlag (2010)  182--194

\bibitem{Goundan2007}
Goundan, P., Schulz, A.:
\newblock {Revisiting the greedy approach to submodular set function
  maximization}.
\newblock Optimization online (2007)  1--25

\bibitem{KM2011}
Apt, K., Markakis, E.:
\newblock Diffusion in social networks with competing products.
\newblock In Persiano, G., ed.: Algorithmic Game Theory. Volume 6982 of Lecture
  Notes in Computer Science.
\newblock Springer Berlin / Heidelberg (2011)  212--223

\bibitem{DBLP:conf/wine/DubeyGM06}
Dubey, P., Garg, R., De~Meyer, B.:
\newblock Competing for customers in a social network: The quasi-linear case.
\newblock In Spirakis, P., Mavronicolas, M., Kontogiannis, S., eds.: Internet
  and Network Economics. Volume 4286 of Lecture Notes in Computer Science.
\newblock Springer Berlin / Heidelberg (2006)  162--173

\bibitem{KostkaOW08}
Kostka, J., Oswald, Y.A., Wattenhofer, R.:
\newblock Word of mouth: Rumor dissemination in social networks.
\newblock In: Proceedings of the 15th international colloquium on Structural
  Information and Communication Complexity. SIROCCO '08, Berlin, Heidelberg,
  Springer-Verlag (2008)  185--196

\bibitem{BBLO12}
Borodin, A., Braverman, M., Lucier, B., Oren, J.:
\newblock Truthful mechanisms for competing submodular processes.
\newblock CoRR \textbf{abs/1202.2097} (2012)

\bibitem{Singer12}
Singer, Y.:
\newblock How to win friends and influence people, truthfully: influence
  maximization mechanisms for social networks.
\newblock In: WSDM. (2012)  733--742

\bibitem{Company1978}
Nemhauser, G.L., Wolsey, L.A., Fisher, M.L.:
\newblock An analysis of approximations for maximizing submodular set
  functions---i.
\newblock Mathematical Programming \textbf{14} (1978)  265--294
  10.1007/BF01588971.

\bibitem{citeulike:340551}
Rockafellar, R.T.:
\newblock {Convex Analysis (Princeton Landmarks in Mathematics and Physics)}.
\newblock {Princeton University Press} (1996)

\bibitem{S04}
Sviridenko, M.:
\newblock A note on maximizing a submodular set function subject to a knapsack
  constraint.
\newblock Operations Research Letters \textbf{32} (2004)  41 -- 43

\bibitem{KST09}
Kulik, A., Shachnai, H., Tamir, T.:
\newblock Maximizing submodular set functions subject to multiple linear
  constraints.
\newblock In: Proceedings of the twentieth Annual ACM-SIAM Symposium on
  Discrete Algorithms. SODA '09, Philadelphia, PA, USA, Society for Industrial
  and Applied Mathematics (2009)  545--554

\bibitem{BorodinBLO13}
Borodin, A., Braverman, M., Lucier, B., Oren, J.:
\newblock Strategyproof mechanisms for competitive influence in networks.
\newblock In: WWW. (2013)  141--150

\bibitem{LSWZ10}
Lu, P., Sun, X., Wang, Y., Zhu, Z.A.:
\newblock Asymptotically optimal strategy-proof mechanisms for two-facility
  games.
\newblock In: Proceedings of the 11th ACM conference on Electronic commerce. EC
  '10, New York, NY, USA, ACM (2010)  315--324

\bibitem{AFKP10}
Ashlagi, I., Fischer, F., Kash, I., Procaccia, A.D.:
\newblock Mix and match.
\newblock In: Proceedings of the 11th ACM conference on Electronic commerce. EC
  '10, New York, NY, USA, ACM (2010)  305--314

\end{thebibliography}
}
\appendix 
\section{Relation with Other Diffusion Models} 
\label{sec.models}
In our results, we have made a number of modelling assumptions about agent utilities and social welfare. To some extent, we can argue that these assumptions may be 
necessary to be able to obtain truthfulness and constant approximation
on the social welfare.  Furthermore, we now provide some background on the relevance of our assumptions to the existing work on influence diffusion in social networks, which served as the running example throughout the paper.

\paragraph{Non-decreasing and submodular utilities and social welfare} To 
the best of our knowledge, in order to establish a constant approximation on the social
welfare, all of the known models in competitive and non-competitive diffusion assume that the overall expected spread is a non-decreasing and submodular function with respect to the set of initial adopters. A main part of the seminal work by Kempe et al. (\cite{Kempe:2003}) is the proof that the expected spread of two models of non-competitive diffusion process is indeed non-decreasing and submodular. This was later extended to more general processes in \cite{Kempe:2005}. In the case of the competitive influence spread models in \cite{DBLP:conf/wine/2007/BharathiKS07}, \cite{Carnes:2007:MIC:1282100.1282167}, and \cite{springerlink:10.1007/978-3-642-17572-5_48}, it is shown that a player's expected spread is a non-decreasing and submodular function of his initial set of nodes, while fixing the competitors allocations of nodes. This also implies that the total influence spread is a non-decreasing and submodular. Without any
assumption on the nature of the social welfare function, it is NP hard to
obtain any non trivial approximation on the social welfare even for a single
player. 

\paragraph{Adverse competition}
In the initial adoption of (say) a technology, a 
competitor can indirectly benefit from competition so as to insure
widespread adoption of the technology. However, once a technology is
established (e.g., cell phone usage), the issue of influence spread 
amongst competitors should satisfy adverse competition. The same can
be said for selecting a candidate in a political election. We also note
that the previous competitive spread models (\cite{DBLP:conf/wine/2007/BharathiKS07}, \cite{Carnes:2007:MIC:1282100.1282167}, and \cite{springerlink:10.1007/978-3-642-17572-5_48}) mentioned above also satisfy adverse competition. In its
generality, the Goyal and Kearns model need not satisfy this assumption, but 
in order to obtain their positive result on the price of anarchy, they
adopt a similar restriction (namely, that the adoption function at
every node satisfies the condition that a player's probability of
influencing an adjacent  node cannot decrease in the absense of other players 
competing).

Furthermore, a simple example shows that the assumption of adverse competition
 is necessary for truthfulness.
Consider the following two-player setting.
The ground set is composed of two items: $u_1$, which contributes a value of $1$ to the receiving player and a value of $N$ to her competitor (who did not receive $u_1$), and item $u_2$ which gives both players a value of $1$.
\comment{Each of the items $u_1,u_2$ and $u_3$ contibutes $N+1$ to the social welfare. Whenever one
of these items is allocated
to a player, it will contribute a value of $1$ to the player to whom it
was allocated,
and will contribute $N$ to its competitor's valuation (thus, the
competitor will get a ``free ride'' as a result of his competitor's
allocation). However, if the competitor's allocation is empty, the entire worth of the item, $N+1$, will
be transfered to the receiving player.}
Now, consider the outcome of any mechanism when the bid profile is
$(1,1)$. Without loss of generality, one player, say player $A$,
will receive $u_1$, while the other player will get $u_2$.
The valuations would therefore be $2$ and $N+1$ for players
$A$ and $B$, respectively. In that case, player $A$ would prefer to
lower her bid to $0$, which would guarantee her a valuation of $N$ (player $B$ would have to get $u_1$, as otherwise the approximation ratio of the social welfare is unbounded as $N$ grows). We conclude that unless the competition assumption holds, no strategyproof mechanism can, in general, obtain a bounded approximation ratio to the optimal social welfare.
%
Although the example refers to deterministic
allocations, the same argument can be made for randomized allocations.

\paragraph{Mechanism and agent indifference} In both the Wave Propagation model and the Distance-Based model presented in \cite{Carnes:2007:MIC:1282100.1282167}, the propagation of influence upholds both the mechanism and agent indifference properties. In \cite{GoyalK12}, it is assumed that the probability that a node will adopt some technology is a function of the fraction of influenced neighbours (regardless of their assumed technology). This immediately implies mechanism indifference, as general spread is invariant with respect to the distribution of technologies among initial nodes. For their positive price of anarchy results about more than two players, it
is assumed that the selection function is linear which would imply
mechanism indifference. 

\paragraph{Anonymity} With the excption of the OR model (\cite{springerlink:10.1007/978-3-642-17572-5_48}), the above mentioned models also satisfy an
anonymity assumption that will be needed to modify the local greedy algorithm 
(as in Algorithms \ref{alg:mech2} and \ref{alg:rand.mech}) 
so as to insure that the 
initial allocation is disjoint (see Appendix \ref{sec.disjoint}). 
Anonymity basically means that the players are symmetric and when there
are more than two players 
this is a somewhat weaker condition than having both mechanism and
agent indifference as we shall show in Appendix \ref{sec.disjoint}. We note that in \cite{DBLP:conf/wine/2007/BharathiKS07} and \cite{Carnes:2007:MIC:1282100.1282167} there is only one edge-weight per edge \footnote{In fact, towards the end of the paper, the authors of \cite{Carnes:2007:MIC:1282100.1282167} conjecture that their results extend to the non-anonymous case where each edge has technology-specific weights. This conjecture was later shown to be false in \cite{springerlink:10.1007/978-3-642-17572-5_48}.} thereby enforcing 
anonymity. In \cite{GoyalK12}, it is explicitly 
stated that the selection function is symmetric across the players and this 
implies anonymity. 

\paragraph{Multiset Allocations and Disjointness}
Our model assumes that each agent can be allocated a node at most once, and indeed most influence models assume that a node is allocated at most once in any initial allocation.  However, we can extend our model to allow a node to be allocated multiple times to the same agent, as in (for example) the model of Goyal and Kearns \cite{GoyalK12}.  To implement such an extension, we can simply consider a modified problem instance with many identical copies of a given node, treating each copy as a distinct element, and then proceed as though each element can be allocated at most once. The output of Algorithm~\ref{alg:mech2} for two players and Algorithm~\ref{alg:rand.mech} for more than two players would then be a profile of multisets with regard to the original network model.  We note, however, that for the case of Algorithm~\ref{alg:rand.mech}, the MeI and AgI definitions effectively imply that if multisets are permitted, then non-disjoint allocations must be permitted as well, as the conditions cannot distinguish between an element being allocated to one agent twice or to two different agents.

\paragraph{Generality of the Model}
A few words are in order about the generality of the model of diffusion under which we
prove that Algorithm \ref{alg:mech2} is strategyproof and provides a 2-approximation. As
noted, with the exception of the OR model, the analysis in previous competitive influence
models assumes anonymous agents. Our general model does not require anonymity and hence we
can accommodate agent specific edge weights (e.g. in determining the probability that
influence is spread across an edge, or for determining whether the weighted sum of
influenced neighbors crosses a given threshold of adoption).  Our model also notably allows
agent-independent node weights, for determining the value of an influenced node.
Moreover, our abstract model does not specify any particular influence spread process, so
long as the social welfare function 
is monotone submodular and each player's payoff is monotonically non-decreasing in his own
set and non-increasing in the allocations to other players.   In particular, our framework
can be used to model probabilistic cascades as well as submodular threshold models.

 \section{Counter examples when there are two agents (extended discussion)}
\label{sec.examples.appendix}
The locally greedy algorithm
is defined over an 
\emph{arbitrary} permutation of the allocation turns. At the core of
our work, we seek to carefully construct such orderings in a manner
that induces strategyproofness.
We demonstrate that this algorithm due to Nemhauser et al 
\cite{Nemhauser1978} (see also 
Goundan and Schultz \cite{Goundan2007})
is not, in general, strategyproof for some natural methods for choosing 
the ordering of the allocation between two agents.  

To clarify the context when there are only two agents, we refer to them 
as agent $A$ and agent $B$ and their utilites as $f_A$ and $f_B$ respectively. 
We give examples of a set $U$ and
functions $f_A$ and $f_B$ (satisfying the conditions of our model) such that
natural greedy algorithms for choosing sets $S$ and $T$ result in
non-monotonicities. 
Our examples  will all easily extend 
to the case of $k > 2$ agents (but not satisfying agent indifference).

\subsection{The OR model}
We will consider examples of a special case of the OR model for influence spread, as studied in
\cite{springerlink:10.1007/978-3-642-17572-5_48}.
Let $G=(V,E)$ be a graph with fractional edge-weights $p: E
\rightarrow [0,1]$, vertex weights $w_v$ for each $v \in V$, and sets $I_A,I_B \subseteq V$
of ``initial adopters'' allocated to each player.
We use vertex weights for clarity in our examples; in Appendix \ref{sec.unweighted} we  
show how to modify the examples given in this section to be unweighted.
The process then unfolds in discrete steps. 
For each $u_A \in I_A$ and $v_A$ such that $(u_A,v_A)
\in E$, $u_A$, once infected, will have a single chance to ``infect'' $v_A$ with probability
$w(u_A,v_A)$. Define the same, single-step process for the nodes in
$I_B$, and let $O_A$ and $O_B$ be the nodes infected by nodes in $I_A$
and $I_B$, respectively. Note that the infection process defined for
each individual player is an instance of the Independent Cascade model
as studied by Kempe et al. \cite{Kempe:2003}. Finally, nodes that
are contained in $O_A \backslash O_B$ will be assigned to player $A$,
nodes in $O_B \backslash O_A$ will be assigned to $B$, and any
nodes in $O_A \cap O_B$ will be assigned to one player or the
other by flipping a fair coin. 

In our examples, we consider two identical players each having 
utility equal to the weight of
the final set of nodes assigned by the spread process. 
It can be easily verified that both the expected social welfare 
(total weight of influenced nodes) and the expected individual 
values (fixing the other player's allocation) 
are submodular set-functions. 

\subsection{Deterministic greedy algorithms that are not strategyproof}
We demonstrate that the more obvious deterministic
orderings for the greedy algorithm fail. First, consider the
``dictatorship'' ordering, in which (without loss of generality
by symmetry) player $A$ is first allocated
nodes according to his budget, and only then player $B$ is allocated nodes. 
Our example showing
non-truthfulness also applies to an ordering that would always 
select the player having the largest current unsatisfied budget 
breaking ties (again without loss of generality by symmetry) 
in favor of player A.
\comment{
We will consider a particular ordering in which the algorithm
first allocates player $A$ items according to its bid, and
subsequently allocates items to player $B$. Our example showing 
non-truthfulness also applies to an ordering that would always first
select the player having the largest current unsatisfied budget say
breaking ties (without loss of gnerality by symmetry) in favor of player A. 
Recall that this particular ordering will not affect the overall approximation ratio of 2 to the social welfare. In our example we will consider the following competitive spread process, which is a simplified version of the OR process, as studied by \cite{springerlink:10.1007/978-3-642-17572-5_48}.
}
Consider the graph depicted in Figure~\ref{fig:counter1}.
\begin{figure}
  \subfloat[The counter-example for the deterministic
  mechanism with a dictatorship ordering. The initial budget for both players is 1.]{
  $\xymatrix{
        *+[o][F-]{u_1 } & &*+[o][F-]{u_2 } & &*+[o][F-]{u_3 }  &*+[o][F-]{u_4 }\\
        \\
        &*+[o][F-]{c_1 } \ar[uul]_{1} \ar[uur]^{1 } & & &*+[o][F-]{c_2 } \ar[uu]_{\frac{1}{4}+\epsilon} \ar[uullll]_{\frac{9}{10}} \ar[uull]_{\frac{9}{10}}  &*+[o][F-]{c_3 } \ar[uu]_{\frac{1}{2}}
      }$
   \label{fig:counter1}
  }\hspace{\stretch{1}}
\subfloat[The counter-example for the deterministic algorithm under a
Round Robin ordering. The initial budgets for players $A$ and $B$ are
$1$ and $2$, respectively.]{
    $\xymatrix{
        *+[o][F-]{u_1 } & &*+[o][F-]{u_2 }  &*+[o][F-]{u_3 }  &*+[o][F-]{u_4 }\\
        \\
        *+[o][F-]{c_1 } \ar[uu]_{1} & &*+[o][F-]{c_2 }
        \ar[uull]_{1 - 2 \cdot \epsilon} \ar[uu]_{ \epsilon } &
          *+[o][F-]{c_3 } \ar[uu]_{1} & *+[o][F-]{c_4 }
          \ar[uu]_{3\cdot \epsilon}
      }$
 \label{fig:counter2}
  }
\caption{Counter-examples for the mechanism under the deterministic
  dictatorship and Round Robin orderings. In both case, we set the
  weights $w_{c_i}=\epsilon$ and $w_{u_i}=1$, for all $1 \leq i \leq 4$. $0 <\epsilon < \frac{1}{8}$}
\end{figure}
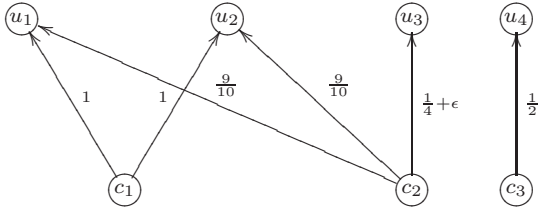
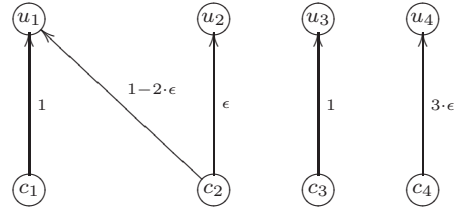
When player $A$ bids 1 and player $B$ bids 1 as well, the algorithm
will allocate $c_1$ to player $A$, as it contributes the maximal
marginal gain to the social welfare, and will allocate $c_3$ to player $B$. 
The value of the allocation for player $A$ is $2$.

However, notice that if player $A$ increases
its bid to 2, the mechanism will
allocate nodes $c_1$ and $c_4$ to player $A$, and allocate
$c_2$ to $B$.  In this case player $A$ receives an extra value of
$\frac{1}{2}$ from node $c_3$, but the allocation of $c_2$ to $B$ will 
``pollute'' player $A$'s value from $c_1$: he will receive 
nodes $u_1$ and $u_2$ each with probability $\frac{1}{10} + \frac{1}{2}\cdot\frac{9}{10} = \frac{11}{20}$.
Thus the total expected value for player $A$ is only $\frac{16}{10}$, and hence the algorithm is
non-monotone in the bid of player $A$.

\comment{
Again, we use the same model in which each
node that was allocated to either one of the players infects its
neignbours with some probability, according to the edge weights. Nodes
that were infected by both players' nodes are ultimately infected by
just one of the players with equal probability.
}

\comment{
In our example, there are items that 
improve the payoff of the receiving player solely by their ability to
potentially ``steal'' from the other player's outcome. If player $A$
were to increase her budget, it would encourage the mechanism to
assign player $B$ more items that steals $A$'s outcome. 

With this in
mind, 
}
Next, consider the Round Robin ordering, in which the mechanism
alternates between allocating a node to player $A$ and to player
$B$. Our example here also applies to the case when the mechanism
always chooses the player having the smallest current unsatisfied
budget breaking ties in favor of player A. 
Consider the instance given in Figure~\ref{fig:counter2}. When
the bids of players $A$ and $B$ are $1$ and $2$, respectively, the
algorithm will first allocate $c_1$ to player $A$, and then it will
subsequently allocate nodes $c_3$ and $c_4$ to player $B$, which
results in a payoff of $1$ for player $A$. If player
$A$ were to increases his bid to $2$, then the mechanism would
allocate nodes $c_1$ and $c_4$ to player $A$, and nodes $c_2$ and
$c_3$ to player $B$, for a payoff of $3 \cdot \epsilon + 2 \cdot
\epsilon + (1 - 2 \cdot \epsilon) \cdot \frac{1}{2} =
\frac{1}{2}+4 \cdot \epsilon < 1$ (since $0 < \epsilon <
\frac{1}{8}$). Therefore, the monotonicity is violated for the
payoff to player $A$.

\subsection{The uniform random greedy algorithm is not strategyproof}
\label{sec.counterexample.random}
As we shall see in
Section~\ref{sec:3player}, for the case of $k>2$ agents in the 
restricted setting that assumes mechanism and agent indifference, 
a very simple mechanism admits a strategyproof
mechanism that provides an
$\frac{e}{e-1}$ approximation to the optimal social welfare. More specifically, we show
that under these assumptions on the social welfare 
agent utilities, taking a uniformly random permutation over the
allocation turns is a strategyproof algorithm. In contrast, for the
case of $k=2$, and even with these additional restrictions (although the
agent indifference assumption turns out to be vacuous in
this case), the uniformly random mechanism is not strategyproof.

Consider the example given in
Figure~\ref{fig:counter3}. 
We note that for this example, Algorithm \ref{alg:rand.mech} 
in Section ~\ref{sec:3player} is equivalent 
to first choosing a random order of allocation (e.g. choosing all 
possible permutations satisfying agent demands with equal
probability) and then allocating greedily. 
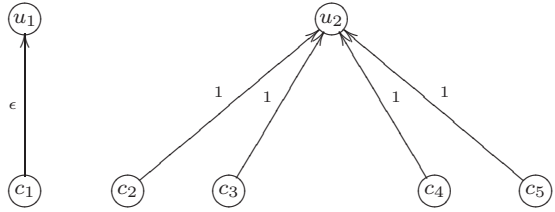
\begin{figure}
\begin{equation*}
\xymatrix{
*+[o][F-]{u_1 } & & &*+[o][F-]{u_2 }\\
\\
*+[o][F-]{c_1 } \ar[uu]^{\epsilon} &*+[o][F-]{c_2 } \ar[uurr]^{1}
&*+[o][F-]{c_3 } \ar[uur]^{1} & & *+[o][F-]{c_4 } \ar[uul]_{1} &*+[o][F-]{c_5 } \ar[uull]_{1}
}
\end{equation*}
\caption{The counterexample for the
  mechanism that allocated according to a random ordering of the turns
  ($0 <\epsilon \ll 1$). $w_{c_i}=\epsilon, i=1,\ldots,5$, $w_{u_i}=1, i=1,2$}
\label{fig:counter3}
\end{figure}
The greedy algorithm will allocate one of $c_2,c_3,c_4$
and $c_5$ to one of the players, then allocate $c_1$, and then any
remaining nodes. 

Let player $A$'s budget be $3$ and  player $B$'s budget be $1$. 
In this case, with
probability $\frac{1}{4}$, player $B$ will be allocated $c_1$ (i.e.\ when $B$'s allocation occurs second), 
in  
which case player $A$'s expected value would be
$1$. Also, with probability $\frac{3}{4}$, player $B$ will
be allocated one of $\{c_2, c_3, c_4, c_5\}$, in which case player $A$'s expected outcome
would be $\frac{1}{2} + \epsilon$. In total, player $A$'s expected payoff will be
$\frac{5}{8} + \frac{3}{4}\epsilon$.

If player $A$ were to increase his budget to $4$, then with
probability $\frac{1}{5}$ player $B$ will be allocated $c_1$, in which
case player $A$'s outcome will be $1$. On the other hand, player $A$'s
expected payoff will be $\frac{1}{2} + \epsilon$ if $B$ is allocated one
of $\{c_2, c_3, c_4, c_5\}$, which occurs with probability $\frac{4}{5}$. In total, player
$A$'s expected outcome will be $\frac{3}{5} + \frac{4}{5}\epsilon < \frac{5}{8} + \frac{3}{4}\epsilon$,
implying that this algorithm is non-monotone.

\section{Counterexamples with Unweighted Nodes}
\label{sec.unweighted}

In Section \ref{sec.examples} we constructed specific examples of influence spread instances for the OR model, to illustrate that simple greedy methods are not necessarily strategyproof for the case of two players.  These examples used weighted nodes which our model allows. For the sake of completeness, we now show that these examples can be extended to the case of unweighted nodes.

We focus on the example from Section \ref{sec.counterexample.random} to illustrate the idea; the other examples can be extended in a similar fashion.  In that example there were nodes $u_1$ and $u_2$ of weight $1$, and nodes $c_1, \dotsc, c_5$ of weight $0$.  We modify the example as follows.  We choose a sufficiently large integer $N > 1$ and a sufficiently small $\epsilon > 0$.  We will replace node $u_1$ with a set $S$ of $2/\epsilon$ independent nodes.  We replace the $\epsilon$-weighted edge from $c_1$ to $u_1$ with an $\epsilon$-weighted edge from $c_1$ to each node in $T$.

Similarly, we replace $u_2$ by a set $T$ of $N$ independent nodes.  For each $c_i$, we replace the unit-weight edge from $c_i$ to $u_2$ with a unit weight edge from $c_i$ to each node in $T$.

In this example, if the sum of agent budgets is at most $5$, the greedy algorithm will never allocate any nodes in $S$ or $T$.  The allocation and analysis then proceeds just as in Section \ref{sec.counterexample.random}, to demonstrate that if agent $B$ declares $1$ then agent $A$ would rather declare $3$ than $4$.


\section{Tightness of Approach: More than Two Players}
\label{sec.3player}

The mechanism we construct in Section \ref{sec.main} is applicable to settings in which there are precisely two competing players, and our mechanism in Section \ref{sec:3player} for more than three players requires the MeI and AgI assumptions.  A natural open question is whether these results can be extended to the general case of three or more agents without the MeI and AgI restrictions.  In this section we briefly describe the difficulty in applying our approach to settings with three players.

For the case of two players in Section \ref{sec.main}, our mechanism was built from an initial greedy algorithm by randomizing over orderings under which to assign elements to players.  Our construction is recursive: we demonstrated that if we can define the behaviour of a strategyproof mechanism for all possible budget declarations up to a total of at most $t$, then we can extend this to a strategyproof mechanism for all possible budget declarations that total at most $t+1$.  A key observation that makes this extension possible is the direct relation between the utilities of the two players. This manifests itself in the cross monotonicities that we utilize in the inductive argument. In addition, the strategyproofness condition (i.e. agent monotonicity) can be equivalently re-expressed as a certain 
``budget competition'' property: if one player increases his budget, then the expected utility for the other player cannot increase by more than the marginal gain the total welfare.  In other words, for all $a + b \geq 1$,  
a strategyproof mechanism must satisfy 
$w^A(a,b) -  w^A(a,b-1) \leq  \Delta^{\oplus B}(a,b)$ where $\Delta^{\oplus B}(a,b) = w(a,b) - w(a,b-1)$ and a similar consequence with regard to 
$\Delta^{\oplus A}(a,b)$. 

\begin{claim}[Equivalence of monotonicity and budget competition for two players]
\begin{itemize}
\item[1:]
$w^A(a,b) - w^A(a,b-1) \leq w(a,b) - w(a,b-1)$ iff $w^B(a,b-1) \leq w^B(a,b)$ \\
\item[2:]
$w^B(a,b) - w^B(a,b-1) \leq w(a,b) - w(a,b-1)$ iff $w^A(a,b-1) \leq w^A(a,b)$ 
\end{itemize}
\end{claim}

\begin{proof}
We provide the proof for $\Delta^{\oplus B}(a,b)$. \\
$w^B(a,b-1) \leq w^B(a,b)$ iff 
$w^A(a,b) + w^B(a,b-1) \leq w^A(a,b) + w^B(a,b)$ iff \\ 
$w^A(a,b) + w^B(a,b-1) + w^A(a,b-1) \leq w^A(a,b) + w^B(a,b) + w^A(a,b-1)$ iff \\
$w^A(a,b) + w(a,b-1) \leq w(a,b) + w^A(a,b-1)$ iff
$w^A(a,b) - w^A(a,b-1) \leq w(a,b) - w(a,b-1)$ 
\qed

\end{proof}

A direct extension of our approach to three players would involve proving inductively that an allocation rule that satisfies these conditions for all budgets that total at most $t$ can always be extended to handle budgets that total up to $t+1$.  We now give an example to show that this is not the case, even when our underlying submodular function takes a very simple linear form.

Suppose we have three players $A$, $B$, and $C$, and suppose our ground set $U$ contains a single element $g$ of value; all other elements are worth nothing.  The utility for each agent is $1$ if their allocation contains $g$, otherwise their utility is $0$.  In this case, the locally greedy algorithm simply gives element $g$ to the first player that is chosen for allocation; the remaining allocations have no effect on the utility of any player.  Note then that the marginal gain in social welfare is $1$ for the first allocation, and $0$ for all subsequent allocations made by the greedy algorithm.

We now define the behaviour of a mechanism for all budget declarations totalling at most $2$.  Note that the relevant feature of this mechanism is the (possibly randomized) choice of which agent is first in the order presented to the greedy algorithm.  We present this behaviour in the following table.

\begin{center}
\begin{tabular}{c|c|c}
Budgets $(a,b,c)$ & Player selected & Utilities $(w^A, w^B, w^C)$ \\
\hline
$(0,0,0)$ & N/A & $(0,0,0)$ \\
\hline
$(1,0,0)$ & $A$ & $(1,0,0)$ \\
$(0,1,0)$ & $B$ & $(0,1,0)$ \\
$(0,0,1)$ & $C$ & $(0,0,1)$ \\
\hline
$(1,1,0)$ & $A$ & $(1,0,0)$ \\
$(0,1,1)$ & $B$ & $(0,1,0)$ \\
$(1,0,1)$ & $C$ & $(0,0,1)$ \\
\end{tabular}
\end{center}

We note that this mechanism (restricted to these type profiles) is strategyproof, satisfies the budget competition property, and also satisfies the cross-monotonicity properties (i.e.\ in the invariants of Theorem \ref{lem:monotone2}).  However, we claim that no allocation on input $(1,1,1)$ that obtains positive social welfare can maintain the budget competition property.  To see this, note that the budget competition property would imply that $w^A(1,1,1) \leq w^A(1,0,1) + \Delta^{\oplus B}(1,1,1) = w^A(1,0,1) = 0$.  Similarly, we must have $w^B(1,1,1) = w^C(1,1,1) = 0$.  Thus, in order to maintain the budget competition property, our mechanism would have to generate social welfare $0$ on input $(1,1,1)$, resulting in an unbounded approximation factor.  We conclude that there is no way to extend this specific mechanism for budgets totalling at most $2$ to a (strategyproof) mechanism for budgets totalling at most $3$ while maintaining the constant approximation factor of the locally greedy algorithm.

Roughly speaking, the problem illustrated by this example is that the presence of more than two bidders means that a substantial increase in the utility gained by one player does not necessarily imply any limits in the utility of another specific player.  This is in contrast to the case of two players, in which the utilities of the two players are more directly related.  This fundamental difference seems to indicate that substantially different techniques will be required in order to construct strategyproof mechanisms with three or more players. 

A different (and natural) approach would be to employ the solution for two players by grouping all but one player at a time, and running the mechanism for two players recursively. However, this method seems ineffective in our setting, as interdependencies between the players' outcomes can introduce non-monotonicities.
This brings into question whether or not the locally greedy method can be
made strategyproof by some method of randomizing over the order in which allocations are made.

This ``2 vs 3 barrier'' is, of course, not unique to our problem. Many
optimization problems (such as graph coloring) are easily
solvable when the size parameter
is $k = 2$ but
become NP-hard when $k \geq 3$. Closer to our setting, the 2 vs 3 barrier 
has been discussed in recent papers concerning mechanism design without
payments, such as in the Lu et al. \cite{LSWZ10} results
for $k$-facility location. Additionally Ashlagi et al. discussed similar issues (\cite{AFKP10}) in the context of mechanisms for kidney exchange. They show that for $n$ points on the line, 
there is a deterministic (respectively, randomized) strategyproof mechanism for placing $k = 2$ facilities (so as to minimize
the sum of distances to the nearest facility) with approximation ratio
$n-2$ (respectively, 4) whereas for $k = 3$ facilities, they do
not know if there is any bounded ratio for deterministic 
strategyproof mechanisms and the 
best known approximation for randomized strategyproof mechanisms
is $O(n)$. 


\end{document}